\documentclass[english]{amsart}
\pdfoutput=1

\usepackage[T1]{fontenc}
\usepackage[utf8]{inputenc}
\usepackage[english]{babel}
\usepackage{amsmath,amssymb,amsfonts,amsthm}

\usepackage{lmodern}
\usepackage{microtype}
\usepackage{mathtools}  %
\usepackage{booktabs}   %
\usepackage{array}
\newcolumntype{R}[1]{>{\raggedleft\let\newline\\\arraybackslash\hspace{0pt}}p{#1}}

\usepackage[]{csquotes}

\usepackage{tikz}
\usepackage{tikz-cd}
\usepackage{float}

\usepackage[]{url}

\usepackage[noend]{algpseudocode}
\usepackage[shortlabels,inline]{enumitem}
\setlist[description]{labelindent=!,font=\normalfont\itshape,itemsep=0ex,partopsep=0ex}

\usepackage[pdftex,hidelinks]{hyperref}

\usepackage[
  citestyle=numeric,
  bibstyle=numeric,
  backend=bibtex,
  isbn=false,
  doi=false,url=false,
  maxnames=4
]{biblatex}
\bibliography{articles}
\DeclareNameAlias{sortname}{first-last}
\usepackage[ruled]{caption}      %
\floatstyle{ruled}
\newfloat{algo}{tp}{lop}
\floatname{algo}{Algorithm}
\def\AlgoInput{\State \emph{Input}~---~}
\def\AlgoOutput{\State \emph{Output}~---~}

\allowdisplaybreaks[4]

\newtheorem{proposition}{Proposition}
\newtheorem{lemma}[proposition]{Lemma}
\newtheorem{theorem}[proposition]{Theorem}
\newtheorem{corollary}[proposition]{Corollary}
\newtheorem{conjecture}[proposition]{Conjecture}

\theoremstyle{remark}
\newtheorem{example}[proposition]{Example}
\newtheorem{remark}[proposition]{Remark}

\newcommand{\mop}[1]{\operatorname{#1}}
\newcommand{\ud}{\mathrm{d}}
\newcommand{\st}{\ \middle|\ }
\newcommand{\eqdef}{\smash{\ \stackrel{\text{def}}{=}\ }}
\newcommand{\Af}{\mathbb{A}}
\newcommand{\PP}{\mathbb{P}}
\newcommand{\N}{\mathbb{N}}
\newcommand{\Z}{\mathbb{Z}}
\newcommand{\Q}{\mathbb{Q}}
\newcommand{\F}{\mathbb{F}}
\newcommand{\K}{\mathbb{K}}
\newcommand{\cL}{\mathcal{L}}

\newcommand{\cO}{\mathcal{O}}
\newcommand{\cK}{\mathcal{K}}
\newcommand{\xx}{\mathbf{x}}

\newcommand{\thom}{\textrm{hom}}
\newcommand{\tGD}{\mathrm{GD}}
\newcommand{\dop}{{\K\langle \delta \rangle}}

\newcommand{\Syz}{\mathcal{S}}
\newcommand{\tSyz}{\mathcal{S}'}
\newcommand{\red}{\mop{red}}
\newcommand{\rem}{\mop{rem}}
\newcommand{\redgd}{\red^\tGD}
\newcommand{\Wa}{W}

\def\geq{\geqslant}
\def\leq{\leqslant}

\newcommand{\Df}{D_{\!\smash{f}}}
\newcommand{\Cf}{r_{\!\smash{f}}}
\newcommand{\Jf}{{\mop{Jac}f}}
\newcommand{\udf}{\mathrm{d}\kern -.1em f}
\newcommand{\udfw}{\udf \wedge}

\newcommand{\gro}{Gröbner basis}
\newcommand{\gros}{Gröbner bases}

\author{Pierre Lairez}
\title{Computing periods of rational integrals}
\date{January 31, 2015}
\address{Inria Saclay, équipe Specfun, France}
\curraddr{Pierre Lairez --- Fäk. II, Sekr. 3-2 --- Technische Universität zu Berlin --- Straße des 17. Juni 136 --- 10623 Berlin --- Deutschland}
\email{pierre@lairez.fr}
\urladdr{pierre.lairez.fr}

\begin{document}

\begin{abstract}
  A period of a rational integral is the result of integrating, with respect to
  one or several variables, a rational function over a closed path.  This work focuses
  particularly on periods depending on a parameter: in this case the period
  under consideration satisfies a linear differential equation, the
  Picard-Fuchs equation.  I give a reduction algorithm that extends  the
  Griffiths-Dwork reduction and apply it to the computation of Picard-Fuchs
  equations.  The resulting algorithm is elementary and has been successfully
  applied to problems that were previously out of reach.
\end{abstract}

\subjclass[2010]{
Primary 68W30;  %
secondary
14K20, %
14F40, %
33F10%
}

\keywords{
  Integration, periods, Picard-Fuchs equation,
  Griffiths-Dwork reduction, algorithms}

\maketitle

\section*{Introduction}

This work studies periods of rational integrals, that is, the result of the
integration, with respect to one or several variables, of a rational function
over a closed path. I focus especially on the case where the period depends on
a parameter.  The fact that periods depending on a parameter of rational or
algebraic integrals satisfy linear differential equations with polynomial
coefficients has emerged from the work of Euler~\cite[\S7]{Eul33} and his
computation of a differential equation\footnote{$(t-t^3)y''+(1-t^2)y'+ty=0$}
for the perimeter of an ellipse as a function of eccentricity.  Since then,
these differential equations, known as \emph{Picard-Fuchs equations}, have
proven to be useful in numerous domains such as combinatorics~\cite{BouMis10},
number theory~\cite{Beu83} or physics~\cite{MorWal09}. They play also a key
role in mirror symmetry~\cite{Mor92}. Research in computer algebra has devoted
great efforts to provide algorithms for computing integrals and, in particular,
Picard-Fuchs equations.  Nevertheless the practical efficiency of current
methods is not satisfactory in many cases.  One reason might be the high level
of generality of most algorithms, which apply to the integration of general
holonomic functions.  Rational functions are certainly very specific among
holonomic functions, but the numerous applications of Picard-Fuchs equations as
well as the fundamental nature of rational functions make it worth developing
specific methods for them.

\subsection*{The problem}

Let~$R$ be a rational function in the variables~$x_1,\dotsc,x_n$,
denoted~$\xx$, and a parameter~$t$, with coefficients in~$\mathbb C$.
Let~$\gamma$ be a~$n$-cycle in~$\mathbb C^n$, e.g. an embedding of the
sphere~$\mathbb S^n$ in~$\mathbb C^n$, on which~$R$ is continuous when~$t$
ranges over some connected open set~$U$ of~$\mathbb C$.  We can form the following
integral, depending on~$t\in U$,
\begin{equation}\label{eqn:introint}
  P(t) \eqdef \oint_\gamma R(t, \xx) \ud \xx,
\end{equation}
where~$\ud\xx$ stands for~$\ud x_1\dotsm\ud x_n$.
\begin{example}\label{ex:apery}
  For~$t \in \mathbb C$, with~$|t| < 17 -12 \sqrt{2}$  
  \[ \displaystyle \sum_{n = 0}^\infty \sum_{k=0}^n \binom{n}{k}^2 \binom{n+k}{k}^2 t^n = \frac{1}{(2\pi i)^3} \oint_\gamma \frac{\ud x\ud y\ud z}{1-(1-xy)z-txzy(1-x)(1-y)(1-z)}, \]
  where the cycle of integration~$\gamma$ is~$\left\{ (x,y,z)\in \mathbb C^3 \st |x|=|y|=|z|=1/2  \right\}$. 
  This is the generating function of Apéry numbers~\cite{Beu83}.
\end{example}
These integrals, for different cycles~$\gamma$, are called the \emph{periods}
of the integral~$\oint\! R$. It is well-known that~$P(t)$ satisfies a linear
differential equation with polynomial coefficients.  It is a consequence of the
finiteness of the algebraic de~Rham cohomology of~$\Af^n
\setminus V(f)$ with~$\mathbb C(t)$ as base field~\cite{Gro66,Mon72}.
Let~$\cL_{R,\gamma}$ denote the differential operator in~$t$ and~$\partial_t$ which
corresponds to the minimal-order equation of~$P(t)$.  That is to say~$\mathcal
L_{R,\gamma}$ is the non zero operator~$\sum_{k=0}^r a_k(t)\partial_t^k$ with
coprime polynomial coefficients and minimal~$r$, such that
\[ \mathcal L_{R,\gamma}(P) \eqdef \sum_{k=0}^r a_k(t) P^{(k)}(t) = 0. \]
Every linear differential equation for~$P(t)$ translates into an operator
which is a left multiple of~$\cL_{R,\gamma}$.

It often happens that the description of the cycle~$\gamma$ is analytic or
topological, sometimes not even explicit, and, to say the least, unsuitable to
a formal algorithmic treatment.  In fact there is no harm in simply
discarding~$\gamma$: there exists a differential equation satisfied 
by all the periods of~$\oint\! R$.  In other words, there exists an
operator in~$t$ and~$\partial_t$ which is a left multiple of
all~$\cL_{R,\gamma}$.  Let~$\cL_R$ denote the least common left multiple of
the~$\cL_{R,\gamma}$.  The classical result which allows the algorithmic computation
of~$\mathcal L_R$ is that it is the minimal operator~$\cL$ such that
\begin{equation}
  \mathcal L(R) = \sum_{i=1}^n \partial_i(B_i)
  \label{eqn:alg-pf}
\end{equation}
for some rational functions~$B_i$ in~$\mathbb C(t, \xx)$ whose denominators
divide a power of the denominator of~$R$, and where~$\partial_i$
denotes~$\partial/\partial x_i$.  This article presents an algorithm
that compute the operator~$\mathcal L_R$, or at least a left multiple of it.

\begin{example}
  In the case of Example~\ref{ex:apery}, the operators~$\cL_R$ and~$\cL_{R,\gamma}$ both equal
  \[ \cL_R =t^2(t^2-34t+1)\partial_t^3 +3t(2t^2-51t+1)\partial_t^2 +(7t^2-112t+1)\partial_t +  (t-5). \] 
\end{example}

Note that integrals of algebraic functions
are easily translated into integrals of rational functions with one variable
more: if~$W(t,\xx)$ is a function such that~$P(t, \xx, W) = 0$ for
some polynomial~$P$ in~$\mathbb C[t, \xx, y]$, elementary residue calculus
shows that
\[ W(t,\xx) = \frac{1}{2\pi i}\oint_\tau \frac{y \partial_yP}{P} \ud y\]
over some adequate contour~$\tau$ and where~$\partial_y$ denotes the derivation~$\partial /\partial y$, so that
\[ \oint_\gamma W(t, \xx) \ud\xx = \frac{1}{2\pi i}\oint_{\gamma\times \tau} \frac{y \partial_yP }{P} \ud\xx\ud y. \]

\subsection*{Contributions}
Following the principle of the reduction of the pole order, I define a family
of finer and finer reductions~$[\,]_r$, for~$r\geq 1$, that given a rational
function~$R$ in several variables produces another rational function~$[R]_r$
that differs from~$R$ only by a sum of partial derivatives of other rational
functions (Section~\ref{sec:horel}). The first reduction~$[\,]_1$ is the
Griffiths-Dwork reduction (Section~\ref{sec:redgd}).

When applied to the case of periods depending on a parameter, these reductions
can solve Equation~\eqref{eqn:alg-pf}, and hence compute Picard-Fuchs equations
of rational integrals (Section~\ref{sec:ct}).  A major difficulty is to fix
an~$r$ such that the~$r$\textsuperscript{th} reduction map~$[\,]_r$ will be fine enough to ensure the
termination of the algorithm. It is solved by applying a theorem of Dimca
(Section~\ref{sec:dimcathm}).

The new algorithm has been implemented and shows excellent performance
(Section~\ref{sec:impl}). For example, I applied it to compute~137 periods
coming from mathematical physics that were previously out of
reach~\cite{BatKre10} (Section~\ref{sec:cy}).

\subsection*{Reduction of pole order}
The principle of the method originates from Hermite reduction~\cite{Her72}.
It is a procedure for computing a normal form of a univariate function modulo
derivatives.  Hermite introduced his method as a way to compute the algebraic
part of the primitive of a univariate rational function without computing the
roots of its denominator, as opposed to the classical partial fraction
decomposition method.  Let~$[R]$ denote the reduction of a fraction~$R$. It is
defined as follows.  Let~$a/f^q$ be a rational function in~$\mathbb C(x)$,
with~$f$ a square-free polynomial and~$q$ a positive integer. Every fraction
can be written in this way since~$a$ and~$f$ are not assumed to be relatively
prime.  If~$q>1$, then we first write~$a$ as~$u f + v f'$, using the assumption
that~$f$ is square-free, and we observe that
\[ \frac{a}{f^q} = \frac{u + \frac{1}{q-1}v'}{f^{q-1}} - \left( \tfrac{1}{q-1}\frac{v}{f^{q-1}} \right)'. \]
This leads to the following recursive definition of~$[a/f^q]$:
\[ \left[ \frac{a}{f^q} \right] = \left[ \frac{u + \frac{1}{q-1}v'}{f^{q-1}} \right]. \]
When~$q = 1$, the reduction~$[a/f]$ is defined to be~$r/f$, where~$r$ is the
remainder in the Euclidean division of~$a$ by~$f$.  Hermite reduction enjoys
the following properties: it is linear; the fractions~$[R]$ and~$R$ differ only
by a derivative of a rational function; and~$[R]$ is zero if and only if~$R$ is the
derivative of a rational function.

The principle of Hermite reduction gives an efficient way to compute the
Picard-Fuchs equation of univariate integrals~\cite{BosCheChy10}.  Let~$R$ be a
rational function in~$\mathbb C(t, x)$.  Hermite reduction can be performed
without modification over the field with one parameter~$\mathbb C(t)$.  To
compute~$\cL_R$, it is sufficient to compute the reductions~$[\partial_t^k R]$,
for~$k \geq 0$, until finding a linear dependency relation over~$\mathbb
C(t)$
\[ \sum_{k=0}^r a_k(t) [\partial_t^k R] = 0. \]
Then the properties of the Hermite reduction assure that~$\cL_R$
is~$\sum_{k=0}^r a_k(t) \partial_t^k$.
The computations of all the reductions~$[\partial_t^k R]$ is improved
significantly when noting the inductive formula~$\left[ \partial_t^{k+1} R \right] = \left[ \partial_t [\partial_t^k R] \right]$.

With several variables, the construction of a normal form modulo derivatives is
considerably harder than with a single variable.  Nonetheless, as soon as we
obtain such a normal form, it is possible to compute Picard-Fuchs equations as
above, by finding linear relations between the~$[\partial_t^k R]$.

Part~\ref{sec:algo} deals with the construction of the maps~$[\,]_r$
whereas Part~\ref{sec:param} deals with the computation of Picard-Fuchs equations.

\subsection*{Related works}
Several existing algorithms are applicable to the computation of~$\cL_R$. The
reader may refer to~\cite{Chy14} for an extensive survey of ``creative
telescoping'' approaches.  A first family, originating in the work of
Fasenmyer~\cite{Fas47} and Verbaeten~\cite{Ver74}, gave rise to an
algorithm by Wilf and Zeilberger~\cite{WilZei92}, refined by Apagodu and Zeilberger~\cite{ApaZei06},
applicable to proper hyperexponential terms, which includes rational functions. The
idea is to transform Equation~\eqref{eqn:alg-pf} into a linear system
over~$\mathbb C(t)$ by bounding \emph{a priori} the order of a left multiple of~$\cL_R$ and the
degree of the polynomials appearing in the fractions~$B_i$.  While being an
interesting method, especially because it gives \emph{a priori} bounds, the
order of the linear system to be solved is large even for moderate sizes of the input.

Zeilberger's ``fast algorithm''~\cite{Zei91} for hypergeometric summation
is the origin of a different family of algorithms, whose key idea is to reduce the resolution of
Equation~\eqref{eqn:alg-pf} to the computation of rational solutions of systems
of ordinary linear differential equations. Interestingly, Picard used this idea much
earlier in a method for computing double rational integrals~\cite{Pic02a}.
Chyzak's algorithm~\cite{Chy00} and Koutschan's
semi-algorithm~\cite{Kou10}---termination is not proven---belong to this
line  and apply to~$D$-finite ideals in Ore algebras.  Rational functions are a
very specific case.

A last family of algorithms coming from~$\mathcal D$-module theory has given algorithms for
numerous operations on~$\mathcal D$-modules and, in particular, an algorithm by
Oaku and Takayama~\cite{OakTak99} to compute the de Rham cohomology of the
complement of an affine hypersurface, which would allow, in theory, to
compute Picard-Fuchs equations.  It is worth noting that an algorithm to
compute the integration of a holonomic~$\mathcal D$-module does not give as
such an algorithm applicable to our problem: computing the annihilator of a
rational function in the Weyl algebra is far from being an easy task~\cite{OakTak99}.

The domain of application of each of these three families is much larger than
just rational integrals: any comparison with the present algorithm must be
done with this point in mind.

The \emph{guessing} method, or \emph{equation reconstruction}, a totally
different method, applies to the computation of~$\cL_{R,\gamma}$.  It often
happens that beside the integral formula for~$P(t)$ one has a way to compute a
power series expansion.  After computing sufficiently many terms, it is
possible to recover~$\cL_{R,\gamma}$ \emph{via} Hermite-Padé approximants.  It
may be difficult to prove that the operator computed is indeed correct, but not
too hard to get convinced. The simplicity of this method counterbalances a
certain lack of delicacy and justifies its ample use. When the power series
expansion of~$P(t)$ is, for some reason, easy to compute, it can find
Picard-Fuchs equations which are far out of reach of any existing
algorithms~\cite[\emph{e.g.}][]{KauZei11}.  Most of the time, though,
the power series expansion of~$P(t)$ is expensive to compute.  For example, I
am aware of no general method allowing to compute directly the first~$p$ terms
of a diagonal of a rational function in~$n$ variables in less than~$p^n$
arithmetic operations. However, space complexity can be improved~\cite{Met12}.

Picard and Simart have studied the case of simple and double
integrals of algebraic functions and gave methods to compute normal forms
modulo derivatives extensively~\cite{PicSim97}\nocite{PicSim06}.
\Citeauthor{CheKauSin12}~\cite{CheKauSin12} gave an algorithm in this
direction, for double rational integrals.  This algorithm is an echo,
independently discovered, of one of the methods of Picard~\cite{Pic02a}.
Interestingly, it has two steps: a first one based on a reduction \emph{à la}
Hermite and another one based on creative telescoping.

Well later after Picard, Griffiths resumed the search for a normal form in the
setting of de Rham cohomology of smooth projective hypersurfaces, defining what
is now known as the Griffiths-Dwork
reduction~\cites[\S3]{Dwo62}[\S8]{Dwo64}[\S4]{Gri69}.  This reduction is in
many respects similar to the Hermite reduction.  It can be applied to the
computation of Picard-Fuchs equations in the same way as Hermite reduction
applies to univariate integrals.  The smoothness hypothesis can be worked around
with a generic deformation. This leads to an interesting complexity result
about the computation of Picard-Fuchs equations~\cite{BosLaiSal13} but to
disappointing practical efficiency in singular cases.  The direction of
Griffiths and Dwork was extended, in particular, by
Dimca~\cite{Dim90,Dim91} and Saito~\cite{DimSai06}, and some
results are known in the case of a singular hypersurface.

\subsection*{Acknowledgment}
I am grateful to Alin Bostan and Bruno Salvy for their precious help and
support, to Mark van~Hoeij and Jean-Marie Maillard for their expertise with
differential operators and to the referee for his thorough work.

\part{Reduction of periods}
\label{sec:algo}

Let~$\K$ be a field of characteristic zero, and let~$A$ be the polynomial
ring~$\K[x_0,\dotsc,x_n]$, for some integer~$n$.  Let~$f$ be an homogeneous
element of~$A$ and let~$A_f$ be the localized ring~$A[1/f]$.
The degree of~$f$ is denoted~$N$.
We focus here on integrals~$\oint R\ud \xx$ which are homogeneous of degree zero, this means that
\[ R(\lambda x_0,\dotsc,\lambda x_n) \ud( \lambda x_0) \dotsm \ud( \lambda x_n) = R(x_0,\dotsc,x_n)\ud x_0\dotsm \ud x_n, \]
or equivalently that~$R$ is a homogeneous rational function of degree~$-n-1$.
Every integral can be homogenized with a new variable,
see~\S\ref{sec:homogenization}.

This part addresses the problem of finding an algorithm \emph{à la} Hermite
that computes an idempotent linear map~$R \mapsto [R]$, from~$A_f$ to itself
such that~$[R]$ equals zero if and only if~$R$ is in the linear
subspace~$\sum_{i=0}^n \partial_i A_f$.  This problem is solved by the Hermite
reduction when~$n$ is~$1$ and by the Griffiths-Dwork reduction when~$f$
satisfies an additional regularity hypothesis (see
Theorems~\ref{thm:griffiths-simple} and~\ref{thm:gd-red}).  To this purpose, a
family of maps, denoted~$[\,]_r$, is constructed such that~$[\,]_1$ is the
Griffiths-Dwork reduction and such that~$[\,]_{r+1}$ factors through~$[\,]_r$.
I give an efficient algorithm to compute these maps.
Conjecturally,~$[\,]_{n+1}$ satisfies the desired properties.  Fortunately,
other results allow to avoid relying on this conjecture when dealing with
periods depending on a parameter.

\section{Overview}
\label{sec:overview}

\subsection{Griffiths-Dwork reduction}

To achieve a normal form modulo derivatives, the guiding principle is the
\emph{reduction of pole order}.  Let us first consider the decision
problem: given a rational function~$a/f^q$, decide whether it lies in~$\sum_{i=0}^n
\partial_i A_f$.  A major actor of the study is~$\Jf$, the Jacobian ideal
of~$f$. It is the ideal of~$A$ generated by the partial derivatives~$\partial_0
f, \dotsc, \partial_n f$.  The basic observation is that the differentiation formula
\begin{equation}
  \sum_{i=0}^n \partial_i\left( \frac{b_i}{f^{q-1}} \right) = \frac{\sum_{i=0}^n \partial_i b_i}{f^{q-1}} - (q-1) \frac{\sum_{i=0}^n b_i\partial_i f}{f^{q}}
  \label{eqn:difffrac}
\end{equation}
implies, by reading it right-to-left, that if~$a\in\Jf$ and~$q > 1$ then~$a/f^q$
equals~$a'/f^{q-1}$ modulo derivatives, for some~polynomial~$a'$.
Namely, if~$a=\sum_i b_i\partial_i f$ then
\[ \frac{a}{f^q} \equiv \frac{\frac{1}{q-1}\sum_{i=0}^n \partial_i b_i}{f^{q-1}} \mod{ \sum_{i=0}^n \partial_i A_f }.\]

Griffiths~\cite{Gri69} proved the converse property in the case when~$\Jf$ is
zero-dimensional or, equivalently, when the projective variety defined by~$f$
is smooth.  Under this hypothesis, if~$q>1$ and if~$a/f^q\equiv a'/f^{q-1}$,
modulo derivatives, for some polynomial~$a'$, then~$a\in\Jf$.  This gives an
algorithm to solve the decision problem, by induction on the pole order~$q$.

\subsection{Singular cases}

In presence of singularities, Griffiths' theorem always fails.
For example, with~$f$ equal to~$xy^2-z^3$,
\begin{equation}
  \frac{x^3}{f^2} =
    \partial_x \left(\frac{\frac 27 x^4}{f^2}\right)
    - \partial_y\left( \frac{\frac17 x^3 y}{f^2} \right), 
  \label{eqn:simplesyz}
\end{equation}
but~$x^3$ is not in~$\Jf$, which is here the ideal~$(xy, y^2, z^2)$.
This identity is a consequence of the following particular case of Equation~\eqref{eqn:difffrac}:
\begin{equation}
  \sum_{i=0}^n b_i\partial_i f = 0 \Rightarrow \sum_{i=0}^n \partial_i\left( \frac{b_i}{f^{q}} \right) = \frac{\sum_{i=0}^n \partial_i b_i}{f^{q}}.
  \label{eqn:syzfrac}
\end{equation}
Tuples of polynomials~$(b_0,\dotsc,b_n)$ such that~$\sum_{i=0}^n b_i\partial_i
f$ are called \emph{syzygies} (of the sequence~$\partial_0f,
\dotsc,\partial_nf$).  Therefore, in order to complete the reduction of pole order
strategy, we should not only consider elements of the Jacobian ideal, but also
elements of the form~$\sum_i \partial_i b_i$, where~$(b_b,\dotsc,b_n)$ is a syzygy.
Such elements are called \emph{differentials of syzygies}.

Considering differential of syzygies is not always enough.
For example, with~$f$ equal to~$x_0^4x_1-x_0^2x_1x_2^2+x_0x_2^4$:
\[ 
  \frac{x_1^7}{f^2} = \frac{1062347}{276480} \frac{89x_0^2+96x_0x_1+712x_2^2}{f} + \sum_{i=0}^2\partial_i\left(\frac{b_i}{f^3}\right),
\]
for some lengthy polynomials~$b_i$, whereas~$x_1^7$ is not a sum of a differential of a syzygy and of an element of~$\Jf$.
Note the exponent~$3$ appearing in~$\partial_i(b_i/f^3)$, it is the least possible.

\subsection{Higher order relations}

Let~$M_q$ be the set of rational functions of the form~$a/f^q$.
Let~$W_q^1$ be the subset of~$M_q\times M_{q-1}$ defined by
\[ W_q^1 = \left\{ \left( (q-1)\frac{\sum_{i=0}^n b_i\partial_i f}{f^q}, \frac{\sum_{i=0}^n \partial_ib_i}{f^{q-1}} \right) \st b_i \in A \right\}. \]
An element $(R, R')$ of~$W_q^1$ relates a rational function~$R$ with a pole
order at most~$q$ with another rational function~$R'$, with pole order at
most~$q-1$, which is equivalent to~$R$ modulo derivatives.  The following
statement is a rewording of Griffiths' result:
\begin{theorem}[Griffiths]\label{thm:griffiths-simple}
  Assume that~$V(f)$ is smooth. For all~$R$ in~$M_q$, homogeneous of degree~$-n-1$, the following assertions are equivalent:
  \begin{enumerate}
    \item $R$ is in~$\sum_i\partial_i A_f$;
    \item there exists~$R'$ in~$M_{q-1} \cap \sum_i\partial_i A_f$ such that~$(R, R')$ is in~$W_q^1$.
  \end{enumerate}
\end{theorem}

The starting point of the method in the general case is to observe that~$W_q^1$
contains ordered pairs in the form~$(0,R')$. Namely, if~$b_0,\dotsc,b_n$ is a
syzygy, then~$(0,\sum_{i}\partial_i b_i / f^{q-1})$ is in~$W_q^1$.  They seem
to be useless relation in view of Theorem~\ref{thm:griffiths-simple}.  However,
for all such pairs~$(0,R')$, the rational function~$R'$ is
in~$\sum_{i}\partial_i A_f$, since it is equivalent to~$0$ modulo
derivatives.

But it is possible, as remarked above, that~$R'$ is not part of a pair~$(R', R'')$ in~$W_{q-1}^1$.
This motivates the definition of~$W_q^2$ as
\[ W_q^2 \eqdef W_q^1 + \left\{ (R, 0) \st (0, R) \in W_{q+1}^1 \right\}. \]
Of course, this can be iterated:
\[ W_q^{r+1} \eqdef W_q^r + \left\{ (R, 0) \st (0, R) \in W_{q+1}^r \right\}. \]
The basic property that is preserved through this induction is that for
all~$(R, R')$ in~$W_q^r$, the first element~$R$ has a pole of order at most~$q$
and is equivalent, modulo derivatives, to the second element~$R'$, which has a
pole of order at most~$q-1$.  When~$V(f)$ is smooth, then~$W_q^r = W_q^1$, for
all~$q$, but when~$V(f)$ is singular, the spaces~$W_q^r$, with~$r>1$, bring new
relations.
This construction is somehow exhaustive. The first result is the following, with no assumption on~$V(f)$:
\begin{theorem}\label{thm:hored-simple}
  There exists an integer~$r \geq 1$, depending only on~$f$, such that for
  all~$q$ and all~$R$ in~$M_q$, homogeneous of degree~$-n-1$, the following
  assertions are equivalent:
  \begin{enumerate}
    \item $R$ is in~$\sum_i\partial_i A_f$;
    \item there exists~$R'$ in~$M_{q-1}\cap\sum_i\partial_i A_f$ such that~$(R, R')$ is in~$W_q^r$.
  \end{enumerate}
\end{theorem}

The algorithm presented in this article is based on this theorem.  The
definition of the~$W_q^r$ gives readily an algorithm to compute these spaces:
it is only a matter of linear algebra. The second result is a method to achieve
efficiency. The two main ingredients are the use of \gros, and the computation
of a basis of \emph{non-trivial} syzygies to catch most elements of~$W_q^r$
at reasonnable cost.

\subsection{Trivial syzygies}
\label{sec:overview:syz}

The space~$W_q^2$ is made from~$W_q^1$ and elements in the
form~$(\sum_{i}\partial_i b_i / f^{q}, 0)$, where~$b_0,\dotsc,b_n$ is a syzygy,
that is~$\sum_i b_i \partial_i f$ vanishes.

Among syzygies, the \emph{trivial syzygies} do not bring new relations to the
relations already in~$W_r^1$.  A syzygy~$b_0,\dotsc,b_n$ is called \emph{trivial}
if there exist polynomials~$c_{i,j}$, with~$c_{i,j}=-c_{j,i}$, such that
\[ b_i = \sum_{j=0}^n c_{i,j}\partial_jf. \]
The antisymmetry property implies that this defines a syzygy, and we check that
\[ \sum_{i=0}^n \partial_ib_i = \sum_{j=0}^n \left( \sum_{i=0}^n \partial_i c_{ij} \right) \partial_j f + \underbrace{\sum_{i,j=0}^n c_{i,j}\partial_i\partial_j f}_{=0}, \]
so that~$\sum_{i=0}^n \partial_ib_i$ is in the Jacobian ideal.
Moreover
\[ \sum_{j} \partial_j \left( \sum_i \partial_i c_{ij} \right) = 0. \]
It follows that the ordered pair~$(\sum_{i}\partial_i b_i / f^{q}, 0)$ is already
in~$W_q^1$.  Thus, in order to compute~$W_q^2$, one may discard trivial syzygies.
Quantitatively, the trivial syzygies are numerous among the syzygies---see, for
example, Table~\ref{tab:dimsyz}---so that discarding them is a tremendous
improvement.  A basis of \emph{non-trivial} syzygies can be computed
efficiently by means of \gros.

\subsection{Reduction procedure}

Let~$R = a/f^q$ be a fraction in~$M_q$. The reduced form~$[R]_r$ is defined by
induction on~$q$ in the following way.  We decompose~$R$ as~$R' + S$ where~$R'$
is minimal in some sense and where~$S$ is the first element of a pair~$(S, T)$
in~$W_q^r$.  Then~$[R]_r$ is defined to be~$R'+[T]_r$.  By
construction~$[R]_r \equiv R$ modulo derivatives.  The constraint on the
homogeneity degree of~$R$ will ensure that~$T$ is zero at some point of the
induction.

\section{Exponential isomorphism}

The \emph{exponential isomorphism}, Theorem~\ref{thm:expiso},
allows to manipulate polynomials rather than
rational functions. It is folklore, for an account see~\cite{Dim91}.  We
work in a homogeneous setting and  we deal only with homogeneous fractions~$R$
of degree~$-n-1$. (So that~$R\ud x_0 \dotsm\ud x_n$ is homogeneous of
degree~$0$.) A fraction~$a/f^q$ is therefore represented solely by its
numerator~$a$: if~$a/f^q$ is homogeneous of degree~$-n-1$, the numerator~$a$ is
homogeneous of degree~$q\deg f -n-1$, so that~$q$ may be recovered from~$a$.
To the usual partial derivative~$\partial_i$ on the rational side corresponds
the \emph{twisted} derivative on the polynomial side
\[ \partial'_ia \eqdef \partial_i a - (\partial_i f) a = e^f \partial_i (a e^{-f}). \]
The exponential isomorphism relates, on the one hand, homogeneous
fractions~$a/f^q$ of degree~$-n-1$ modulo derivatives and, on the other hand,
homogeneous polynomials, with degree in~$(\deg f)\Z -n-1$, modulo twisted
derivatives.

\subsection{Differential forms}
\label{sec:diffforms}

This section is a short reminder about differential forms, or simply
\emph{forms}.\footnote{See, for example \cite[chap.~10]{Mat80} and \cite[\S
10]{A3}, for more general and complete definitions.}
Let~$\Omega^1$ denote the polynomial differential~$1$-forms: it is the
free~$A$-module of rank~$n+1$, and the basis is denoted by the symbols~$\ud
x_0,\dotsc,\ud x_n$.  The differential map~$\ud$ from~$A$ to~$\Omega^1$ is
defined by
\[ \ud a = \sum_{i=0}^n \partial_i a \,\ud x_i. \]
The~$A$-algebra of differential forms, denoted~$\Omega$, is the exterior
algebra over~$\Omega^1$.  Its multiplication is denoted~$\wedge$, it is generated by
the~$\ud x_i$ and is subject to the relations~$\ud x_i \wedge \ud x_j = -\ud
x_j \wedge\ud x_i$.  The~$A$-module of~$p$-forms, denoted~$\Omega^p$, is the
submodule of~$\Omega$ generated by the~$\ud x_{i_1}\wedge\dotsb\wedge\ud
x_{i_p}$.  With the multi-index notation, this is denoted~$\ud x^I$,
with~$I=(i_1,\dotsc,i_p)$.  $\Omega^p$ is a free module of rank~$\binom{n}{p}$.
The module of~$0$-forms~$\Omega^0$ is identified with~$A$.  As a
module,~$\Omega$ decomposes as~$\oplus_{p=0}^n\Omega^p$.  Specifically, the
module~$\Omega^{n+1}$ has rank~$1$ and is freely generated by~$\ud
x_0\wedge\dotsb\wedge\ud x_n$, denoted~$\omega$.  The module~$\Omega^{n}$ has
rank~$n+1$ and is freely generated by the elements~$\xi_i$ defined by
\[ \xi_i \eqdef (-1)^i \ud x_0\wedge\dotsb\wedge\ud x_{i-1}\wedge\ud x_{i+1}\wedge\dotsb\wedge\ud x_n. \]

\subsubsection{Exterior derivative}
The differential map~$\ud$, from~$A$ to~$\Omega^1$, extends to an endomorphism
of~$\Omega $, called \emph{exterior derivative}, such that for~$\alpha \in
\Omega^p $ and~$\beta\in\Omega $,
\[ \ud(\alpha\wedge\beta) = \ud \alpha \wedge \beta + (-1)^p \alpha \wedge \ud \beta. \]
In particular~$\ud(\Omega^p)$ is included in~$\Omega^{p+1}$ and~$\ud^2 = 0$.
For a~$n$-form~$\beta$, written as~$\sum_{i} b_i \xi_i$, we check
that~$\ud \beta$ equals~$\left(\sum_i \partial_i b_i \right)\omega$.
The exterior derivative gives rise to a complex
\[ 0\longrightarrow A \stackrel{\ud}\longrightarrow \Omega^1 
  \stackrel{\ud}\longrightarrow \dotsb \stackrel{\ud}\longrightarrow
  \Omega^{n}  \stackrel{\ud}\longrightarrow \Omega^{n+1}  \longrightarrow 0
\]
which is exact.

\subsubsection{Homogeneity}
The degree of a monomial~$x^I \ud x^J$ is defined to be~$|I|+|J|$.  A form is
called \emph{homogeneous of degree~$k$} if it is a linear combination of monomials of degree~$k$.
If~$\alpha$ and~$\beta$ are two homogeneous forms of degree~$k_\alpha$
and~$k_\beta$ respectively, then~$\ud \alpha$ is a homogeneous form of
degree~$k_\alpha$ and~$\alpha\wedge\beta$ is a homogeneous form of
degree~$k_\alpha+k_\beta$.

\subsubsection{Koszul complex}\label{sec:kcomplex}
The exterior product with~$\udf$
gives a map from~$\Omega^p$ to~$\Omega^{p+1}$, and since~$\udf\wedge\udf$
vanishes we can consider the chain complex
\[ \cK(\udf) : 0\longrightarrow  A \stackrel{\udf}\longrightarrow \Omega^1 
  \stackrel{\udf}\longrightarrow \dotsb \stackrel{\udf}\longrightarrow
  \Omega^{n}  \stackrel{\udf}\longrightarrow \Omega^{n+1}  \longrightarrow 0,
\]
known as the \emph{Koszul complex} of~$A$ with respect to~$\udf$, and its
cohomology~$H \cK(\udf)$ defined by
\[ H^p \cK(\udf) = \frac{ \Omega^{p} \cap \ker \udf }{ \udf\wedge\Omega^{p-1} }. \]
For a~$n$-form~$\beta$, written as~$\sum_{i} b_i \xi_i$, the exterior
product~$\udf\wedge\beta$ is~$\left(\sum_i b_i \partial_i f \right)\omega$.
Thus~$H^{n+1} \cK(\udf)$ is isomorphic to~$A/\Jf$, with a shift of~$n+1$ in the
natural grading, where~$\Jf$ is the Jacobian
ideal~$(\partial_0 f,\dotsc,\partial_nf)$.

Let~${\rm Syz}$ be the kernel of the product by~$\udf$ on~$\Omega^n$. It is the syzygy module
of the sequence~$\partial_0 f, \dotsc, \partial_nf$.  Let~${\rm Syz}'$
be~$\udfw\Omega^{n-1}$, the module of trivial syzygies, generated by the
elements~$\partial_i f \:\xi_j - \partial_j f \:\xi_i$. In particular~$H^n \cK(\ud f)$ is~${\rm Syz}/{\rm Syz'}$.

\subsection{Chain complex~$T^p$}\label{sec:setup}

For an integer~$q$, let~$T^{p}_q$ be the subspace of~$\Omega^p$ generated by
the homogeneous elements of degree~$qN$.  Let~$T^p$ be the direct sum~$\oplus_q
T^p_q$  and let~$F_q T^p$ be~$\oplus_{q'\leq q} T^{p}_{q'}$.  Note
that~$\udfw$ maps~$T^{n}_q$ to~$T^{n+1}_{q+1}$ and that~$\ud$ maps~$T^{n}_q$
to~$T^{n+1}_q$.  Let~$\Syz$ (resp.~$\tSyz$) be the intersection of~$T^n$
and~${\rm Syz}$ (resp.~${\rm Syz}'$). The component of degree~$qN$ of an
element~$\alpha$ of~$T$ is denoted~$\alpha_q$.

The space~$T^{n+1}_q$ is the equivalent of~$M_q$, as defined in the
introductory remarks: the elements of~$T^{n+1}_q$ represent numerators of
rational functions whose denominator is~$f^q$. We define the linear map~$h$
from~$T^{n+1}$ to~$A_f$ by
\[ h : a\omega \in T^{n+1}_q \longmapsto (q-1)! \frac{a}{f^q} \in A_f. \]
Of course~$h$ is not injective since~$h(f \alpha) = q h(\alpha)$,
for~$\alpha\in T^{n+1}_q$.    Finally let~$\Df$, the \emph{twisted differential}, from~$T^{p}$ to~$T^{p+1}$ be
the map defined by~$\Df\alpha = \ud \alpha - \udfw\alpha$.  Note that~$\Df$
maps~$F_q T^{n}$ to~$F_{q+1}T^{n+1}$.  The anticommutation~$\ud(\udfw \beta) =
-\udfw\ud\beta$ ensures that~$\Df\circ\Df = 0$, so that~$T^p$ forms a chain
complex.

\begin{remark}
  The spaces~$T_q^{p+q}$ arranged within a grid form a double
  complex, known as \emph{Rham-Koszul double complex}~\cite{Dim92}, with the
  \emph{horizontal} differential being~$\ud$ and the \emph{vertical} one
  being~$\udfw$,  see Figure~\ref{fig:rk-cpx}. This arrangement may help
  visualize some of the proofs in this article.
  
  \begin{figure}[tp]
  \[  \begin{tikzcd}
       T_{q+1}^n \arrow{r}{\ud}  & T^{n+1}_{q+1} \arrow{r}{\ud} & 0 & {} \\
       \arrow{r} & T^n_q \arrow{r}{\ud}\arrow{u}{\udf} & T^{n+1}_q \arrow{r}{\ud} \arrow{u}{\udf} & 0 \\
       \arrow{r} & T^{n-1}_{q-1} \arrow{r}{\ud}\arrow{u}{\udf}  & T_{q-1}^n \arrow{r}{\ud}\arrow{u}{\udf}  & T_{q-1}^{n+1} \arrow{u}{\udf} \\
  &\arrow{u} &\arrow{u} &\arrow{u} & 
    \end{tikzcd}
  \]
  \caption{Rham--Koszul double complex}
  \label{fig:rk-cpx}
  \end{figure}
\end{remark}

For~$p \geq 0$, let~$H^p_{\text{Rham}}(\PP^n_\K\setminus V(f))$ be the~$p$\textsuperscript{th}
de~Rham cohomology group of the variety~$\PP^n_\K\setminus V(f)$,\footnote{See~\cite{Gro66} for a general definition and~\cite[]{Gri69} for a definition in the specific case of a complement of a projective hypersurface}
and let~$H^{p+1} T$ be the~$p$\textsuperscript{th} cohomology group of the complex~$T$, that is~$(T^{p}\cap\ker \Df) / \Df(T^{p-1})$.
The following theorem has been proved in numerous occasions under several appearances, it goes back at least to Dwork.
In this exact form, I am aware of proofs by Dimca~\cite[Theorem~1.8]{Dim90},
Malgrange~\cite{Mal91a} and Deligne~\cite{Del91}.

\begin{theorem}
  $H^{p+1} T \simeq H^p_{\text{Rham}}(\PP^n_\K\setminus V(f))$, for all~$p \geq 1$.
  \label{thm:expiso}
\end{theorem}

We will only make use of Theorem~\ref{thm:expiso} in the case where~$p=n$.  The
cohomology group~$H^{n+1} T$ is~$T^{n+1}/\Df(T^n)$
and~$H^n_{\text{Rham}}(\PP^n_\K\setminus V(f))$ is isomorphic to the subspace
of~$A_f/\sum_i\partial_i A_f$ generated by the homogeneous elements of
degree~$-n-1$, and the isomorphism is the map induced by~$h : T^{n+1}
\to A_f$:
\begin{proposition}
  $h(\Df(T^{n})) \subset \sum_{i=0}^n \partial_i A_f$. In other words, the
  map~$h$ induces a map from~$T^{n+1}/\Df(T^n)$ to~$A_f/\sum_i\partial_i A_f$.
  
  \label{prop:hD}
\end{proposition}
\begin{proof}
  Let~$\beta = \sum_{i=0}^n b_i\xi_i$ be an element of~$T^{n}_q$, then
  \begin{align*}
    h\left( \Df\left( \smash{\sum_{i=0}^n} b_i \xi_i \right) \right) &= \sum_{i=0}^n h(\partial_i b \omega) - h(b_i\partial_if\omega)
    = \sum_{i=0}^n (q-1)! \frac{\partial_ib_i}{f^q} - q! \frac{b_i\partial_if}{f^{q+1}}\\
    &= (q-1)! \sum_{i=0}^n \partial_i\left( \frac{b_i}{f^q} \right). \text{\qedhere}
  \end{align*}
\end{proof}

This way, the goal of computing normal forms modulo derivatives of rational
functions can be reformulated as computing normal forms of elements
of~$T^{n+1}$ modulo~$\Df(T^n)$.

\begin{example}
  With~$f = x^2y - z^3$, Equation~\eqref{eqn:simplesyz} rewrites
  \[ x^3 \ud x \ud y \ud z = \Df( \tfrac{2}{7} x^4 \ud y \ud z + \tfrac 17 x^3 \ud x\ud z ). \]
  The rewriting is not always as simple as in this example but Theorem~\ref{thm:expiso}
  asserts that it is always possible.
\end{example}

\subsection{Filtered maps} The space~$T^{n+1}$ admits a filtration given by
the subspaces~$F_q T^{n+1}$ with $q\in\Z$.  In the next sections, we will define
reduction maps which will be \emph{filtered} endomorphisms of~$T^{n+1}$, that is to
say linear maps~$u : T^{n+1}\to T^{n+1}$ such that~$u(F_qT^{n+1})\subset
F_qT^{n+1}$ for all~$q\in\Z$.
Two filtered endomorphisms of~$T^{n+1}$, say~$u$ and~$v$, are \emph{equivalent} if for all~$q\in\Z$ and all~$\alpha\in F_qT^{n+1}$ we have~$u(\alpha) \equiv v(\alpha)$ modulo~$F_{q-1}T^{n+1}$.

For all filtered map~$u$, we can define the \emph{associated graded map} as
\[ \mop{Gr}u : \alpha \in T^{n+1} \mapsto \sum_{q\geq 0} u(\alpha_q)_q \in T^{n+1}. \]
Two filtered maps are equivalent if and only if their associated graded maps are equal.

\section{Griffiths-Dwork reduction}
\label{sec:redgd}

We reword the Griffiths-Dwork reduction, presented in
Section~\ref{sec:overview}, in the above setting.  Let us choose a
monomial ordering on~$A$, denoted~$\prec$.  For a linear subspace~$V$ of~$A$ and
an element~$a$ of~$A$, let~$\mop{rem}(a,V)$, be the unique~$b$
in~$A$ such that~$a-b$ is in~$V$ and no monomials of~$b$ is divided by the
leading monomial of some element of~$V$.  If~$V$ is an ideal of~$A$, this can
be computed using a \gro\ of~$V$, and if it is a finite-dimensional subspace,
then Gaussian elimination following the monomial ordering computes~$\mop{rem}(a,V)$.

The elementary step of the Griffiths-Dwork reduction is the following.
Let~$\alpha$ be an element of~$T^{n+1}_q$.  By definition there is a~$\beta$
in~$T^n$ such that~$\alpha = \rem(\alpha,{\udfw T^n}) + \udfw \beta$. We choose~$\beta$ in such a way that:
   it depends linearly on~$\alpha$ ;
  $\beta = 0$ if~$\alpha$ is in~$\Df T^n$ ;
   and $\beta$ is in~$T^n_{q-1}$.
The \emph{elementary reduction of~$\alpha$ in degree~$q$} is then defined to
be
\begin{equation}
   \redgd_q(\alpha) \eqdef \mop{rem}(\alpha,{\udfw T^n}) + \ud\beta.
  \label{equ:defredq}
\end{equation}
For~$\alpha$ in~$T_k^{n+1}$, for some~$k$ different from~$q$, we
define~$\redgd_q(\alpha) = \alpha$.  The definition of~$\redgd_q$ depends on the
choice of~$\beta$; however, the equivalence class of~$\redgd_q$ as a filtered
map does not.

This elementary reduction is very easy to compute using a \gro\ of the Jacobian
ideal~$\Jf=(\partial_0f,\dotsc,\partial_nf)$ and its cofactors.  Indeed, the
multivariate division algorithm gives a decomposition of a polynomial~$a$
as~$\mop{rem}(a,\Jf) + \sum_{i=0}^n b_i \partial_i f$.  If~$\alpha$
is~$a\omega$, then~$\mop{rem}(\alpha,\udfw T^n)$ is~$\mop{rem}(a,\Jf)\omega$
and~$\beta$ may be chosen equal to~$\sum_i b_i\xi_i$.  In this way, the
assumptions on~$\beta$ are naturally satisfied. See Section~\ref{sec:impl} for
more details about the implementation.

By construction, $\alpha - \redgd_q(\alpha) = -\Df\beta$, so that~$\redgd_q$ is
an idempotent map whose kernel is included in~$\Df(T^n)$.  When translated into
a relation between fractions, this reflects integration by parts:
\[ \oint b_i \partial_i( 1/f^{q-1} ) \ud\xx = -\oint \partial_i b_i / f^{q-1} \ud \xx. \]

This reduction step can be iterated and for~$\alpha\in F_q T^{n+1}$, the
\emph{Griffiths-Dwork reduction} of~$\alpha$, denoted $[\alpha]_{\tGD}$, is
defined as
\[ [\alpha]_{\tGD} \eqdef \redgd_1 \circ \dotsm \circ\redgd_q(\alpha). \]
We check the following recursive relation: $[\alpha]_{\tGD} = \mop{rem}(\alpha,\udfw T^n) + [\ud\beta]_\tGD$, 
where~$\beta$ is the one in the equation~\ref{equ:defredq}.  Again, the
map~$[\,]_\tGD$ depend on the choice of~$\beta$ but its equivalence class, as a
filtered map, does not.

\begin{proposition}
  \label{prop:elem-red}
  The map~$[\,]_\tGD$ is filtered, idempotent and~$\ker [\,]_\tGD \subset
  \Df(T^n)$.  In particular~$\alpha \equiv [\alpha]_{\tGD}$ modulo~$\Df(T^{n})$
  for all~$\alpha\in T^{n+1}$.  Moreover, for all~$q\geq 0$ and~$\alpha\in
  F_qT^{n+1}$, $[\alpha]_\tGD \in F_{q-1}T^{n+1}$ if and only if $\alpha\in
  \Df(F_{q-1}T^n) + F_{q-1}T^{n+1}$.
\end{proposition}

\begin{proof}
  It is straightforward that the map~$[\,]_\tGD$ is filtered, idempotent and
  that~$\ker [\,]_\tGD$ is included in~$\Df(T^n)$.  Concerning the second
  point, let~$\alpha\in F_qT^{n+1}$.  By contruction~$[\alpha]_\tGD =
  \alpha + \Df\beta$ for some~$\beta\in F_{q-1}T^n$.  So if~$[\alpha]_\tGD$ is
  in~$F_{q-1}T^{n+1}$ then~$\alpha$ is in~$\Df(F_{q-1}T^n) + F_{q-1}T^{n+1}$.

  Conversely, let~$\alpha = \Df\beta + \varepsilon$, with~$\beta\in
  F_{q-1}T^{n}$ and~$\varepsilon\in F_{q-1}T^{n+1}$.  Then~$\alpha_q =
  \udfw\beta_{q-1}$, and so~$\rem(\alpha_q, \udfw T^n) = 0$.  By
  Equation~\eqref{equ:defredq}, $\redgd_q(\alpha)\in F_{q-1}T^{n+1}$, and
  so~$[\alpha]_\tGD$ as well.
\end{proof}

The Griffiths-Dwork reduction~$[\,]_{\tGD}$ is a multivariate and homogeneous
analogue of Hermite reduction.  In general, it does not have all the nice
properties of Hermite reduction: it may happen that for some~$\alpha$
in~$\Df(T^{n})$ the reduction~$[\alpha]_{\tGD}$ is not zero and it may fail at 
reducing the degree.  Nevertheless, Dwork~\cites[\S3]{Dwo62}[\S8]{Dwo64} and Griffiths~\cite[\S4]{Gri69} have proven the
following:
\begin{theorem}[Dwork, Griffiths]
  If~$V(f)$ is smooth in~$\PP_\K^n$ then
  \begin{enumerate}[(i)]
    \item $\ker [\,]_{\tGD} = \Df(T^{n})$,
    \item for all~$\alpha$ in~$T^{n+1}$ the reduction~$[\alpha]_{\tGD}$ is in~$F_nT^{n+1}$
  \end{enumerate}
  \label{thm:gd-red}
\end{theorem}

\begin{remark}\label{rem:equ-redgd}
This theorem still holds if we replace~$[\,]_\tGD$ by any equivalent filtered
 and idempotent map~$u$ whose kernel is included in~$\Df(T^n)$.
Indeed, in this case~$\ker \mop{Gr}u = \ker \mop{Gr}[\,]_\tGD = \mop{Gr}(\Df
T^n)$ Since~$\ker \mop{Gr}u$ is~$\mop{Gr}(\ker u)$, this implies that~$\ker u =
\Df(T^n)$.  Moreover, the point~(ii) implies that~$[F_q T^{n+1}]_\tGD \subset
F_{q-1}T^{n+1}$ for all~$q > n$.  Since~$[\,]_\tGD$ and~$u$ are equivalent, the
same holds for~$u$.  And since~$u$ is assumed to be idempotent, this implies
that~$u(T^{n+1})\subset F_n T^{n+1}$.
\end{remark}

The hypothesis ``$V(f)$ is smooth'' is equivalent to the fact that~$\Jf$ is a
zero-dimensional ideal, that is~$A/\mop{Jac}f$ is finite-dimensional over~$\K$.
It is also equivalent to the equality of~$\Syz$ and~$\tSyz$, respectively the
syzygies and the trivial syzygies in~$T^n$.  The main step of the proof of
Theorem~\ref{thm:gd-red} is~\cite[Theorem~4.3]{Gri69}:
\begin{theorem}[Dwork, Griffiths]
  If~$V(f)$ is smooth in~$\PP_\K^n$ then~$\Df(T^n)\cap F_q T^{n+1}$ is
  contained in~$\Df(F_{q-1}T^n)$ for all~$q\geq 0$.
  \label{thm:grif-ordcert}
\end{theorem}

In the singular case, it is never true that~$\ker [\,]_{\tGD} =
\Df(T^{n})$.  Worse still, 
the cokernel~$T^{n+1}/\ker [\,]_{\tGD}$ is never finite dimensional.  Indeed, we
have
\[ \frac{F_{q} T^{n+1}}{ F_q T^{n+1} \cap \ker[\,]_{\tGD} + F_{q-1}T^{n+1}} \simeq (A/\Jf)_{qN-n-1}, \]
so the quotient is finite dimensional if and only if~$\Jf$ is a zero-dimensional ideal.

\section{Computation of higher order relations}
\label{sec:horel}

\subsection{Construction}

Let~$\Wa_q^1$ be the subspace of~$T^{n+1}_q+T^{n+1}_{q-1}$ defined by
\begin{equation}
  \Wa_q^1 \eqdef \Df(T_{q-1}^{n} ) = \left\{  -\udfw\beta + \ud\beta \st \beta \in T^n_{q-1} \right\}.
\end{equation}
Following the idea developed in Section~\ref{sec:overview}, we define, for~$r\geq 1$ and~$q\geq0$
\[ \Wa_q^{r+1} \eqdef \Wa_q^1 + \Wa_{q+1}^r \cap F_q T^{n+1}. \]
Compared to Section~\ref{sec:overview}, the space~$M_q$ has been replaced
by~$T_q^{n+1}$ and the product~$M_{q}\times M_{q-1}$ by the direct
sum~$T_{q}^{n+1} \oplus T_{q-1}^{n+1}$.

\begin{figure}[tp]
  \[  \begin{tikzcd}
      0 \mathrlap{\ ( = -\udfw\beta_{q+2})} & & & & {} \\
      \beta_{q+2} \arrow{r}{\ud}\arrow{u}{-\udf} & 0 \mathrlap{\ ( = \ud\beta_{q+2} - \udfw\beta_{q+1})}& & & \\
      {} & \beta_{q+1} \arrow{r}{\ud}\arrow{u}{-\udf} & 0 & & \\
      {} & & \beta_{q} \arrow{r}{\ud}\arrow{u}{-\udf} & \alpha_q & \\
      {} & &  & \beta_{q-1} \arrow{r}{\ud}\arrow{u}{-\udf} & \alpha_{q-1}
    \end{tikzcd}
  \]
  \caption{A~$n$-form~$\beta\in T_{q-1}^n + \dotsb + T_{q+2}^n$ such
  that~$\Df\beta \in F_q T^{n+1}$, thus giving an element~$\alpha$
  of~$W_q^4$.}
  \label{fig:wrq}
\end{figure}

\begin{proposition}
  For all~$r\geq 1$ and~$q \geq 0$,
  \[ \Wa_q^r = \Df\big(\sum_{k=1}^r T_{q+k-2}^n\big) \bigcap F_q T^{n+1}. \]
\label{prop:caracW}
\end{proposition}

\begin{proof}
  By induction on~$r$. For~$r = 1$, the claim reduces to~$W_q^1 =
  \Df(T_{q-1}^{n})$, which is the definition.  Then, let us prove that the
  right-hand side satisfies the recurrence relation defining~$W_q^r$, that is:
  \[ \Df\big(\sum_{k=1}^{r+1} T_{q+k-2}^n\big) \bigcap F_q T^{n+1} = \Df(T_{q-1}^{n}) + \Df\big(\sum_{k=1}^r T_{q+k-1}^n\big) \bigcap F_q T^{n+1}, \]
  which follows simply from~$\Df(T_{q-1}^{n}) \subset F_qT^{n+1}$.
\end{proof}

Figure~\ref{fig:wrq} depicts what are elements of~$W_q^r$.

\begin{example}
  With~$f = x y^2 - z^3$, we find that~$W_1^1 = 0$ and
  \[ W_2^1 = \left\langle x^2y, xy^2, y^3, xyz, y^2z, xz^2, yz^2, z^3, 1 \right\rangle\omega. \]
  Thus~$W_2^1 \cap T^3_1=\langle \omega \rangle$ and~$W_1^2=\langle \omega \rangle$.
\end{example}

\subsection{Reductions of order~$r$}

The higher order analogue of~$\redgd_q$, denoted~$\red_q^r$ is  the linear
map~$T^{n+1}\to T^{n+1}$ defined by
\[ \red_q^r \alpha \eqdef \mop{rem}(\alpha, W_q^r), \]
for~$\alpha$ in~$T_q^{n+1}$, and~$\red_q^r \alpha = \alpha$ for~$\alpha\in
T_k^r$ with~$k\neq q$.  As for the Griffiths-Dwork reduction, we define
for~$\alpha$ in~$F_q T^{n+1}$
\[ [\alpha]_r \eqdef \red_1^r \circ \dotsb\circ \red_q^r(\alpha). \]
This reduction map enjoys the following properties, to be compared with
Proposition~\ref{prop:elem-red} relative to~$[\,]_\tGD$.

\begin{proposition}
  Let~$r\geq 1$. The map~$[\,]_r$ is filtered and idempotent, its kernel is
  included in~$\Df(T^n)$ and~$[\,]_{r+1}\circ[\,]_r = [\,]_{r+1}$.  Moreover, for all~$q \geq 0$ and~$\alpha \in F_q
  T^{n+1}$, $[\alpha]_r \in {F_{q-1} T^{n+1}}$ if and only if~$\alpha \in \Df(
  F_{q+r-2}T^{n+1} ) + F_{q-1}T^{n+1}$.
  \label{thm:caracredr}
\end{proposition}

\begin{proof}
  It is straightforward that the map~$[\,]_r$ is filtered and idempotent.
  Since~$W_q^r\subset \Df(T^n)$, for all~$q$, we have~$\ker [\,]_r\subset \Df(T^n)$.
  And since~$W_q^{r} \subset W_q^{r+1}$ we have~$[\,]_{r+1}\circ[\,]_r=[\,]_{r+1}$.

  Let~$\alpha\in F_q T^{n+1}$ such that~$[\alpha]_r \in {F_{q-1} T^{n+1}}$.
  From the definition, $[\alpha]_r \equiv \red_q^r \alpha \pmod{F_{q-1} T^{n+1}}$
  and~$\red_q^r \alpha \equiv \alpha \pmod{W_q^r}$.
  So~$\alpha \equiv 0 \pmod{W_q^r + F_{q-1}T^{n+1}}$ and~$\alpha \in \Df( F_{q+r-2}T^{n+1} ) + F_{q-1}T^{n+1}$.

  Conversely, let us assume that~$\alpha = \Df\beta + \alpha'$, with~$\beta$
  in~$F_{q+r-2}T^{n+1}$ and~$\alpha'$ in~$F_{q-1}T^{n+1}$.  The form~$\beta$
  splits as~$\beta' + \varepsilon$, with~$\beta' \in \sum_{k=1}^r T_{q+k-2}^{n}$
  and~$\varepsilon\in F_{q-2}T^n$.
  We check that~$\Df\beta' \in F_qT^{n+1}$, so~$\Df\beta'\in W_q^r$, by Proposition~\ref{prop:caracW}.
  And~$\red_q^r(\Df\beta') \in F_{q-1}T^{n+1}$, by definition of~$\red_q^r$.
  Thus
  \[ \red_q^r(\alpha) = \red_q^r(\Df\beta') + \red_q^r(\Df\varepsilon + \alpha')  \in F_{q-1}T^{n+1}, \]
  and~$[\alpha]_r$, which equals~$[\red_q^r(\alpha)]_r$, is in~$F_{q-1}T^{n+1}$ as well.
\end{proof}

\begin{corollary}\label{coro:exhaus-red}
  $\Df(T^n) = \bigcup_{r\geq1} \ker[\,]_r$.
\end{corollary}
\begin{proof}
  Let~$\beta\in T^n$ such that~$\Df\beta \neq 0$. Let~$q\geq 0$ be the least integer such that~$\Df\beta \in F_qT^{n+1}$.
  Let~$r\geq 1$ such that~$\beta\in F_{q+r-2} T^{n+1}$.
  By Proposition~\ref{thm:caracredr}, $[\Df\beta]_r$ is in~$F_{q-1}T^n$, and it is also in~$\Df(T^n)$ because~$\Df\beta\equiv[\Df\beta]_r$
  modulo~$\Df(T^n)$.
  By induction on~$q$, there exists an~$s\geq r$ such that~$\left[ [\Df\beta]_r \right]_s = 0$.
  Since~$[\,]_s\circ[\,]_r = [\,]_s$, the result follows.
\end{proof}

\begin{remark}
  The reductions~$[\,]_\tGD$ and~$[\,]_1$ do not necessarily coincide, but they
  are equivalent filtered maps.
\end{remark}

Thus, we have a family of finer and finer reductions which generalize the
Griffiths-Dwork reduction and which are exhaustive in the sense that they
reduce to zero every~$\Df\beta$ if~$r$ is large enough.  However, two problems
remains. The first on is practical: as defined, the computation of~$[\,]_r$,
for a given~$r$, involves the resolution of huge linear systems, both when
computing the spaces~$W_q^r$ and when computing~$\red_q^r$.  This is in
contrast with~$[\,]_\tGD$ which only involve the computation of a \gro\ and
reductions modulo it for computing~$\redgd_q$. The~\S\ref{sec:fastcomp}
describe a faster way to compute~$[\,]_r$.  The second problem is theoretical:
how to set the parameter~$r$?  This is addressed in
Section~\ref{sec:dimcathm}.

\subsection{Faster computation}\label{sec:fastcomp}
There are two ingredient for computing~$[\,]_r$ faster than with plain linear
algebra.  The first is the use of~$\redgd$, whose implementation is efficient
and which readily perform a great deal of reductions.  Secondly, we discard
trivial syzygies, as explained in~\S\ref{sec:overview:syz}.

Let~$A_q$ be a complementary subspace of~$\tSyz_q$ in~$\Syz_q$, that
is~$\Syz_q$ equals~$\tSyz_{q} \oplus A_{q}$.  Let~$X_q^1 \eqdef \ud A_{q-1}$
and,  for all~$q\geq0$ and~$r\geq1$,
\[ X_q^{r+1} \eqdef \ud A_{q-1} + \redgd_q\left( X_{q+1}^r\cap F_q T^{n+1} \right). \]
Since~$\ud A_{q-1} = \Df(A_{q-1})$, it is clear that~$X_q^1 \subset \Df(F_{q-1}
T^n)$, and by induction on~$r$, we obtain that~$X_q^r \subset \Df(F_{q+r-2}
T^{n})$. Moreover, we have~$\redgd_q \alpha = \alpha$ for all~$q$ and all~$\alpha\in X_q^r$.
Finally, let~$\rho_q^r : T^{n+1}\to T^{n+1}$ the linear map defined by
\[ \rho_q^r(\alpha) \eqdef \mop{rem}(\redgd_q(\alpha), X_q^r), \]
for~$\alpha$ in~$T_q^{n+1}$, and~$\rho_q^r(\alpha)  = \alpha$ for~$\alpha\in
T_k^r$ with~$k\neq q$.
For~$\alpha\in F_qT^{n+1}$ we define
\[ [\alpha]_r' \eqdef \rho_1^r \circ \dotsb\circ \rho_q^r(\alpha). \]

This paragraph aims at proving the following:
\begin{theorem}
  For all~$r\geq1$, the map~$[\,]'_r$ is filtered and idempotent, its kernel is
  included in~$\Df(T^n)$ and~$[\,]'_{r+1}\circ[\,]'_r = [\,]'_{r+1}$.  Moreover,
  it is equivalent to~$[\,]_r$, in particular, for all~$q \geq 0$ and~$\alpha
  \in F_q T^{n+1}$, $[\alpha]'_r \in {F_{q-1} T^{n+1}}$ if and only if~$\alpha
  \in \Df( F_{q+r-2}T^{n+1} ) + F_{q-1}T^{n+1}$.
  \label{thm:fastcomp}
\end{theorem}

\begin{corollary}\label{coro:exhausred-fast}
  $\Df(T^n) = \bigcup_{r\geq1} \ker[\,]'_r$.
\end{corollary}
\begin{proof}
  The proof is the same as Corollary~\ref{coro:exhaus-red}.
\end{proof}

The map~$[\,]'_r$ is easier to compute than~$[\,]_r$ because the linear algebra
involved in the computation of~$X_q^r$ arises in much lower dimension than the
one for~$W_q^r$.  It comes at the cost of using~$\redgd$ and of computing the
space~$A_q$ of non trivial syzygies, which can be done efficiently through
\gros computations, see Section~\ref{sec:impl}.

The main fact which allows to discard trivial syzygies is the following:
\begin{lemma}
  $\redgd_q(\ud \tSyz_q) \subset \ud \Syz_{q-1}$, for all~$q\geq 0$.
  \label{lem:redtsyz}
\end{lemma}

\begin{proof}
  Recall that~$\tSyz_q = \udfw T_{q-1}^n$, so let~$\beta\in T_{q-1}^n$.  The
  differential anti-commutes with~$\udfw$ so that~$\ud(\udfw\beta) =
  -\udfw\ud\beta$.  By definition~$\redgd_q(\ud(\udfw\beta))$ is thus~$\ud
  \gamma$ for some~$\gamma\in T_{q-1}^n$ such that~$\udfw\gamma =
  -\udfw\ud\beta$.  Thus~$\gamma = -\ud\beta + \varepsilon$, for
  some~$\varepsilon\in\Syz_{q-1}$. Since~$\ud(\ud\beta) = 0$, we obtain
  that~$\redgd_q(\ud(\udfw\beta)) = \ud\varepsilon$.
\end{proof}

Let~$G_q\subset T^{n+1}$ be the kernel of~$\redgd_q$. It is a subspace
of~$T_q^{n+1}\oplus T_{q-1}^{n+1}$.

\begin{proposition}
  $W_q^r = X_q^r + G_q + \ud \tSyz_{q-1}$, for all~$q\geq 0$ and~$r\geq 1$.
  \label{prop:wrqsimplified}
\end{proposition}

\begin{proof}
  We proceed by induction on~$r$. When~$r=1$, it boils down to proving
  that~$\Df(T_{q-1}^n) = \ud A_{q-1} + G_q + \ud \tSyz_{q-1}$, that
  is~$\Df(T_{q-1}^n) =  G_q + \ud \Syz_{q-1}$, using the fact that~$\ud A_{q-1}
  + \ud \tSyz_{q-1} = \ud\Syz_{q-1}$.  Let~$\beta\in T_{q-1}^n$. By
  definition of~$\redgd_q$,
  \[ \redgd_q(\Df\beta) = - \redgd_q(\udfw\beta) + \ud\beta = \ud(\beta-\beta'), \]
  for some~$\beta'\in T_{q-1}^n$ such that~$\udfw\beta' = \udfw\beta$.
  Thus~$\beta-\beta'$ lies in~$\Syz_{q-1}$ and~$\redgd_{q}(\Df\beta)$ is
  in~$\ud\Syz_{q-1}$.  Moreover, since~$\redgd_q$ is idempotent, $\Df\beta -
  \redgd_{q}(\Df\beta)$ is in~$G_{q}$, and in the end~$\Df\beta \in G_q + \ud
  \Syz_{q-1}$.  Conversely, $\Syz_{q-1}\subset T_{q-1}^{n}$, so it
  remains to prove that~$G_q\subset \Df(T_{q-1}^n)$, which is easy from the
  definitions.

  Now let~$r \geq 1$.
  By definition, and by the induction hypothesis
  \begin{align*}
    W_q^{r+1} &= W_q^1 + W_{q+1}^r \cap F_qT^{n+1} \\
    &= G_q + \ud A_{q-1} + \ud \tSyz_{q-1} + (X_{q+1}^r + \ud \tSyz_q + G_{q+1}) \cap F_{q}T^{n+1}.
  \end{align*}
  And we have
  \begin{align*}
    (X_{q+1}^r + \ud \tSyz_q + G_{q+1}) \cap F_{q}T^{n+1} = X_{q+1}^r\cap F_{q}T^{n+1} + \ud \tSyz_q.
  \end{align*}
  Indeed~$\ud\tSyz_q\subset F_qT^{n+1}$, and if~$\alpha\in X_{q+1}^r$ and~$\alpha'\in G_{q+1}$
  are such that~$\alpha+\alpha'\in F_qT^{n+1}$, then~$\alpha'=0$
  because
  \[ \alpha + \alpha' = \redgd_{q+1}(\alpha+\alpha') = \redgd_{q+1}(\alpha)+\redgd_{q+1}(\alpha') = \alpha + 0. \]
  Thus~$W_q^{r+1} = G_q + \ud A_{q-1} + \ud\tSyz_{q-1} + \ud\tSyz_q + X_{q+1}^r\cap F_{q}T^{n+1}$.
  For any linear subspace~$A\subset T^{n+1}$, the decomposition~$\alpha\in A$
  as~$\redgd_{q}\alpha+(\alpha-\redgd_{q}\alpha)$ shows that~$G_q + \redgd_q(A)
  = G_q + A$.
  Thus
  \[ W_q^{r+1} = G_q + \ud A_{q-1} + \ud\tSyz_{q-1} + \redgd_q(\ud\tSyz_q) + \redgd_q(X_{q+1}^r\cap F_{q}T^{n+1}), \]
  and the statement follows, by Lemma~\ref{lem:redtsyz} and the definition of~$X_q^{r+1}$.
\end{proof}

We may now prove Theorem~\ref{thm:fastcomp}.
\begin{proof}[Proof of Theorem~\ref{thm:fastcomp}]
  It is straightforward that~$[\,]'_r$ is filtered and idempotent, that~$\ker
  [\,]'_r\subset \Df(T^n)$ and that~$[\,]'_{r+1}\circ[\,]'_{r} =
  [\,]'_{r+1}$.

  To prove that~$[\,]_r$ and~$[\,]'_r$ are equivalent, it is enough to prove
  that~$\red_q^r$ and~$\rho_q^r$ are equivalent.  And indeed, if~$\alpha\in
  F_qT^{n+1}$ then
  \begin{align*}
    \rho_q^r(\alpha) &\equiv \rem(\alpha, G_q + X_q^r)  \mod{F_{q-1} T^{n+1}} \\
\text{and}\quad \red_q^r(\alpha) &\equiv \rem(\alpha, \ud \tSyz_{q-1} + G_q + X_q^r) \mod{F_{q-1} T^{n+1}},
  \end{align*}
  using Proposition~\ref{prop:wrqsimplified}.
  Since~$\ud\tSyz_{q-1} \subset F_{q-1}T^{n+1}$ the claim follows.
\end{proof}

In what follows, $[\,]_r$ will stand for~$[\,]'_r$. Except in terms of
computational complexity, they have the same properties.

\subsection{Quantitative facts}\label{sec:examples}
It is useful to introduce the spaces
\[ E_q^r \eqdef \frac{ F_q T^{n+1} } { \Df( F_{q+r-2}T^n ) \cap F_{q}T^{n+1}  + F_{q-1}T^{n+1}}. \]
It is clear that~$E_q^0$ is~$F_{q} T^{n+1}/F_{q-1}T^{n+1}$, which is isomorphic
to~$T_q^{n+1}$. Moreover, as a reformulation of Proposition~\ref{prop:elem-red}, the space~$E_q^1$ is
\[ E_q^1 = \mop{coker}( \mop{Gr}[\,]_\tGD )_q \eqdef \frac{ F_q T^{n+1} }{ \left\{ \alpha \in F_qT^{n+1} \st [\alpha]_\tGD\in F_{q-1}T^{n+1} \right\} }  \simeq  \frac{T_q^{n+1}}{\udfw T_{q-1}^n}. \]
And by Proposition~\ref{thm:caracredr}, this generalizes to the
isomorphism~$E_q^r \simeq \mop{coker}( \mop{Gr}[\,]_r )_q$.  In other
words,~$E_q^r$ is~$F_{q} T^{n+1}$ modulo elements which are reducible
to~$F_{q-1}T^{n+1}$ by~$[\,]_r$ %
The space~$E_q^{r+1}$ is a quotient of~$E_q^r$, and the dimension fall
represents how many new relations in degree~$qN$ are computed by~$[\,]_{r+1}$
compared to~$[\,]_r$. For~$r=2$, we check that
\[ E_q^2 \simeq \frac{T_q^{n+1}}{\udfw T_{q-1}^n + \ud \Syz_q} = \frac{T_q^{n+1}}{\udfw T_{q-1}^n + \ud A_q}. \]

The dimension of~$E_q^0$ is~$\binom{Nq - 1}{n}$, which is equivalent to~$N^n
q^n/n!$ when~$q\to\infty$.  The dimension of~$E_q^1$ is~$\mathcal O(q^{\nu})$,
where~$\nu$ is the dimension of the singular locus of~$V(f)$ in~$\PP^n_\K$.
There is no easy estimate of the dimension of~$E_q^2$, but~$\dim A_{q-1}$ is
also~$\mathcal O(q^\nu)$. By contrast, $\dim \Syz_q \sim (n+1) N^n q^n / n!$.
For the computation of~$[\,]_2$ (or rather~$[\,]'_2$), it is thus a substantial
improvement to consider the non-trivial syzygies~$A_q$ rather than all the
syzgies~$\Syz_{q}$ .

\begin{table}[t]
  \setlength{\extrarowheight}{.2em}
  \centering
  \begin{tabular}{>{\raggedleft}p{6em}*{5}{R{2em}}@{\hspace{2.5em}}r}
  \toprule
  $q$               & 0   & 1   & 2   & 3   & 4 & $q > 4$ \\ \midrule
  $\dim E_q^0$  & 0   & 10  & 165 & 680 & 1771 & $\binom{6q - 1}{3}\sim 36q^3$ \\ 
  $\dim E_q^1$      & 0   & 10  &  86 & 102 & 120 & $18q+48$ \\
  $\dim E_q^2$      & 0   & 10  &  7  & 6   & 6 & 6\\
  $\dim E_q^3$      & 0   & 9   &  1 & 0   & 0 & 0\\
  $\dim E_q^r$, $r\geq 4$      & 0   & 9   &  1 & 0   & 0 & 0\\
  \bottomrule
  \end{tabular}
  \vspace{.5\baselineskip}
  
  \caption{Some dimensions related to Example~\ref{ex:main}}
  \label{tab:dimEqr}
\end{table}

\begin{table}[t]
  \setlength{\extrarowheight}{.2em}    
  \centering
  \begin{tabular}{r*{5}{R{2em}}@{\hspace{2.5em}}r}
    \toprule
    $q$   & 0   & 1   & 2   & 3   & 4 & $q > 4$ \\ \midrule
    $\dim \Syz_q$
          & 0   & 21  & 522 & 2429 & 6604 & $\sim 144 q^3$ \\
    $\dim A_q$
          & 0   & 1   & 92  & 132  & 168  & $36 q + 24$ \\
    $\dim E_q^1 - \dim E_q^2$
          & 0   & 0   & 79  & 96   & 114  & $18 q + 42$ \\ \bottomrule
  \end{tabular}
  \vspace{.5\baselineskip}
  \caption{Gain of dimension by discarding trivial syzygies and number of new relations generated by the syzygies in the Example~\ref{ex:main}}
  \label{tab:dimsyz}
\end{table}

\begin{example}\label{ex:main}  
To illustrate precisely what does bring the maps~$[\,]_r$ in comparaison with~$[\,]_\tGD$,
let us consider the polynomial~$f$
\[ f \eqdef 2x_1x_2x_3(x_0-x_1)(x_0-x_2)(x_0-x_3)-x_0^3(x_0^3-x_0^2x_3+x_1x_2x_3) \]
coming from an integral for the Apéry numbers, see Example~\ref{ex:apery}. In this case~$n=3$ and~$N=6$.
The dimension of the singular locus of~$V(f)$ in~$\PP_\K^3$ is~$1$.

The dimensions of the first few~$E_q^r$ are shown in Table~\ref{tab:dimEqr}.
This illustrates the successive dimension falls.  Noticeably, at~$r=3$ a new
relation appears in~$F_1T^{n+1}$.  It is~$(2x_1^2-2x_2^2- x_0(x_1-x_2))\omega$,
which equals~$\Df\beta$ for some~$\beta$ in~$F_{2}T^{n}$ but no such~$\beta$ is
small enough to be reproduced here.

Illustrating the same polynomial~$f$, Table~\ref{tab:dimsyz} shows the numbers
of syzygies and non-trivial syzygies at a given degree.  It also displays the
difference~$\dim E_q^1 - \dim E_q^2$, that is how many new relations are really
generated from the syzygies.
\end{example}

\section{Extensions of Griffiths' theorems}
\label{sec:dimcathm}

Given~$\alpha$ in~$T^{n+1}$, how can we compute a~$r$ such that if~$\alpha$ is
in~$\Df(T^n)$ then~$[\alpha]_r$ equals zero?  Corollaries~\ref{coro:exhaus-red}
and~\ref{coro:exhausred-fast} are lacking effective bounds and do not answer
this question.  Dimca proved two theorems~\cites[Th.~B and
Cor.~2]{Dim91}[Th.~2.7]{Dim90} which generalize Theorem~\ref{thm:gd-red}.
While they do not give a full answer, they allow to give enough guarantees
on~$[\,]_r$ to design algorithms that terminates.
\begin{theorem}[Dimca]
  There exists an integer~$C$, depending only on~$f$, 
  such that~$\Df(T^n)\cap F_q T^{n+1} \subset \Df( F_{q+C-2} T^n )$ for all~$q\geq 0$. 
  \label{thm:degen}
\end{theorem}
This statement is to be compared with Theorem~\ref{thm:grif-ordcert}.  Given~$f$ and~$q$, it is
easy to prove that there exists a~$C$ such that~$\Df(T^n)\cap F_q T^{n+1}
\subset \Df( F_{q+C-2} T^n )$, because the left-hand side is a finite
dimensional space and it is included in~$\cup_{C\geq 0}\Df( F_{q+C-2} T^n )$.
It is remarkable that one can choose a~$C$ which does not depend on~$q$.
Let~$\Cf$ be the least such~$C$.
\begin{corollary}
  $\ker[\,]_{\Cf} = \Df(T^n)$.
  \label{coro:degen}
\end{corollary}

\begin{proof}
  Let~$\beta \in T^n$ and~$q \geq 0$ the least integer such that~$\Df\beta \in F_qT^{n+1}$.
  By Theorem~\ref{thm:degen}, there exists~$\beta'\in F_{q+\Cf-2}T^{n}$ such that~$\Df\beta' = \Df\beta$.
  Thus, by Theorem~\ref{thm:fastcomp}, $[\Df\beta]_{\Cf}$ is in~$F_{q-1}T^{n+1}$, and besides, it is also in~$\Df(T^{n})$.
  By induction on~$q$, $\left[ [\Df\beta]_{\Cf} \right]_{\Cf} = 0$. Since~$[\,]_{\Cf}$ is idempotent, the claim follows.
\end{proof}

Unfortunately, this integer~$\Cf$,  while
explicit, is not easy to compute: in Dimca's proof it is expressed in terms of
a resolution of the singularities of the projective variety~$V(f)$.
By contrast, the point~(ii) of Theorem~\ref{thm:gd-red} fully generalizes to singular cases:
\begin{theorem}[Dimca]
  $\Df(T^n) + F_n T^{n+1} = T^{n+1}$.
  \label{thm:einfsupp}
\end{theorem}

\begin{corollary}
  For all~$\alpha\in T^{n+1}$, the reduction~$[\alpha]_{\Cf}$ lie in~$F_n T^{n+1}$.
\end{corollary}

\begin{proof}
  By Theorem~\ref{thm:einfsupp}, there exists~$\beta\in T^n$ such that~$\alpha + \Df\beta$ is in~$F_{n}T^{n+1}$.
  Since~$[\alpha]_{\Cf} = [\alpha + \Df\beta]_{\Cf} - [\Df\beta]_{\Cf}$, the claim follows from Corollary~\ref{coro:degen}.
\end{proof}

For some applications, such that the computation of annihilating operators of
periods with a parameter, Theorem~\ref{thm:einfsupp} gives an efficient
workaround to the lack of \emph{a priori} bounds for~$\Cf$. Consider an
algorithm which computes reductions~$[\alpha]_r$, for some forms~$\alpha$ and
some fixed integer~$r$, and does it as long as the reductions it computes are
linearly independent.  Then either all the~$[\alpha]_r$ are in the finite
dimensional space~$F_nT^{n+1}$, and then the algorithm terminates; or
some~$[\alpha]'_r$ is not in~$F_nT^{n+1}$, and then~$r < \Cf$, by
Theorem~\ref{thm:einfsupp}.  When the second case is encountered, we abort the
algorithm, increment~$r$ and start over.  This may happen only if~$r < \Cf$,
and when it happens~$r$ increases. So it may happen only finitely many times
and the algorithm terminates.

Concerning the integer~$\Cf$ Dimca~\cite{Dim91} conjectured that
\begin{conjecture}\label{conj:dimca}
  $\Cf \leq n+1$.
\end{conjecture}

As far as I know, computations on explicit examples confirm this conjecture.
Moreover the bound is tight when~$n=2$.  A proof of this conjecture would
have very interesting algorithmic consequences: the reduction algorithm is
extendable to the computation of the whole cohomology of~$T$, not just the top
cohomology. Only the bound~$\Cf\leq n+1$ is lacking for obtaining an efficient
algorithm for computing the de~Rham cohomology of the complement of a
projective hypersurface.

\part{Periods with a parameter}
\label{sec:param}

We apply the reduction algorithm to the computation of Picard-Fuchs equations.

\section{Algorithms}
\label{sec:ct}

\subsection{Setting}

Let~$\K$ be a field of characteristic zero with a derivation~$\delta$.
Typically~$\K$ is~$\Q(t)$ and~$\delta$ is the usual derivation with respect
to~$t$.  Let~$\dop$ be the algebra of differential operators in~$\delta$: it is 
the associative algebra with unity generated over~$\K$ by~$\delta$ and subject
to the relations~$\delta x = x \delta + \delta(x)$
for all~$x$ in~$\K$, where~$\delta(x)$ denotes the application of~$\delta$ to~$x$
whereas~$\delta x$ is the operator that multiplies by~$x$ and then
applies~$\delta$.
On~$\K(x_0,\dotsc,x_n)$, let~$\partial_i$ denote the derivation with respect
to~$x_i$.  The derivation~$\delta$ extends to~$\K(x_0,\dotsc,x_n)$ uniquely by
setting~$\delta(x_i) = 0$.  In particular~$\delta \circ \partial_i = \partial_i
\circ \delta$.

This section describes an algorithm which takes as input a rational
function~$R$ in~$\K(x_1,\dotsc,x_n)$ and outputs an operator~$\cL$ in~$\dop$
such that there exist other rational functions~$C_1,\dotsc,C_n$ with
\[ \cL(R) = \sum_{i=1}^n \partial_i C_i. \]
Moreover, the irreducible factors of the denominators of the~$C_i$ divide the
denominator of~$R$.  Such an operator will be called an \emph{annihilating
operator of the periods of~$R$}, or \emph{a differential equation for~$\oint R$}.
The minimal annihilating operator of~$\oint R$ is called the \emph{Picard-Fuchs
equation} (of~$\oint R$).  The output operator~$\cL$ is not necessarily the
Picard-Fuchs equation but it is of course a left multiple of it.

Being based on the reduction algorithm of Part~\ref{sec:algo}, the algorithm
does not compute the~$C_i$.  It is worth a word because while only~$\cL$
matters, the size of the~$C_i$, say the size of a binary dense representation,
is usually much larger than the size of~$\cL$~\cite[Rem.~11]{BosLaiSal13}.  To be able to compute~$\cL$
without computing the~$C_i$ is certainly a good point toward practical
efficiency.  The fractions~$C_i$ are called a \emph{certificate}: they allow to
check \emph{a posteriori} that~$\cL$ is indeed an annihilating operator
of~$\oint R$.

\subsection{Homogenization}\label{sec:homogenization}

The reduction algorithm works in an homogeneous setting.  If we are interested
in computing the Picard-Fuchs equation of the integral of an inhomogeneous
function, the problem can be homogenized as follows.  Let~$R_\thom$ be the
homogenization of~$R$ in degree~$-n-1$ defined by
\[ R_\thom = x_0^{-n-1} R\left( \frac{x_1}{x_0},\dotsc,\frac{x_n}{x_0} \right) \in \K(\xx), \]
where~$\xx$ denotes~$x_0,\dotsc,x_n$ hereafter.
The rational function~$R_\thom(\xx)$ is \emph{homogeneous of degree~$-n-1$}, that is
$R_\thom(\lambda x_0, \dotsc, \lambda x_n) = \lambda^{-n-1} R_\thom(x_0, \dotsc, x_n)$,
or, equivalently,~$R_\thom = b/g$ where~$b$ and~$g$ are homogeneous polynomials such
that~$\deg b +n+1 = \deg g$.

Let us write~$R_\thom$ as~$a/f^q$, with~$a$ and~$f$ two homogeneous polynomials
and~$q$ an integer.  Usually~$f$ will be chosen square-free but we don't have to.
Let~$N$ be the degree of~$f$.  Since~$R_\thom$ is homogeneous of
degree~$-n-1$, the degree of~$a$ is~$qN-n-1$.  This is the main reason for
considering homogeneous fractions: the degree of the denominator determines the
degree of the numerator, there is no \emph{hidden} pole at infinity.  The
degree~$-n-1$ is crucial to ensure that:
\begin{lemma}\label{lem:hom}
  If~$\cL \in \dop$ is a annihilating operator of~$\oint R_\thom$ then~$\cL$ is also a
  annihilating operator of~$\oint R$.
\end{lemma}
\begin{proof}
  Assume that~$\cL(R_\thom)$ equals~$\sum_{i=0}^n \partial_i(b_i/f^m)$, for
  some polynomials~$b_i$ and some integer~$m$.  Substituting~$x_0$ by~$1$
  gives
  \[ \cL(R) = \partial_0(b_0/f^m) |_{x_0 = 1} + \sum_{i=1}^n \partial_i(b_i/f^m |_{x_0 = 1}). \]
  Since~$R_\thom$ is homogeneous of degree~$-n-1$, we may assume that
  each~$b_i/f^m$ is homogeneous of degree~$-n$.  Euler's relation gives
  \begin{align*}
    - n b_0/f^m &= \sum_{i=0}^n x_i \partial_i( b_0/f^m ) 
    = \sum_{i=0}^n\left( \partial_i( x_i b_0 /f^m ) - b_0 /f^m \right).
  \end{align*}
  This proves that
  $0 = \partial_0(b_0/f^m) |_{x_0 = 1} + \smash{\sum_{i=1}^n} \partial_i( x_i b_0 /f^m |_{x_0 = 1} )$, 
  and the claim follows.
\end{proof}
The Picard-Fuchs equation of~$\oint R_\thom$ may not be the
Picard-Fuchs equation of~$\oint R$, but only a left multiple.  However, it is the case if~$x_0$
divides~$f$, which is possible to assume, up to replacing~$f$ by~$x_0 f$
and~$a$ by~$x_0^q a$.  From now on I focus exclusively on the homogeneous
case.

\subsection{Computation of Picard-Fuchs equations}

\begin{algo}
\centering
\begin{algorithmic}
  \AlgoInput $a/f^q$ a homogeneous rational function in~$\K(\mathbf x)$ of degree~$-n-1$, with~$V(f)$ smooth in~$\PP_\K^n$
  \AlgoOutput $\cL \in \dop$ the Picard-Fuchs equation of~$\oint R$
  \Procedure{PicardFuchs}{$a/f^q$}
    \State $\rho_0 \gets [a\omega]_{\tGD}$
    \For{$m$ from~$0$ to~$\infty$}
      \If{$\mop{rank}_\K(\rho_0,\dotsc,\rho_m) = m+1$}
        \State $\rho_{m+1} \gets [\delta(\rho_{m})]_{\tGD}$ 
      \Else
        \State compute~$a_0, \dotsc, a_{m-1} \in \K$ such that~$\sum_{k=0}^{m-1} a_k \rho_k = \rho_m$
        \State \Return $\delta^m - \sum_{k=0}^{m-1} a_k \delta^k$
      \EndIf
    \EndFor
  \EndProcedure
\end{algorithmic}
\caption{Computation of annihilating operators of the periods of a rational function, smooth case}
\label{algo:ctgd}
\end{algo}

The derivation~$\delta$ is extended to the spaces~$T^p$ of differential
forms\footnote{See definition in~\S\ref{sec:setup}.} by
\[ \delta : \alpha \in T^p \mapsto \alpha^\delta - f^\delta\alpha \in T^p, \]
where~$\bullet^\delta$ denotes component-wise differentiation.
It commutes with the map~$h$, and the differential~$\Df$, as a consequence
of~$\delta$ commuting with~$\partial_i$. 

To highlight the difference between the smooth and the singular cases, I recall
first how the Griffiths-Dwork reduction applies to the computation of
Picard-Fuchs equations.  Let~$a/f^q$ be a homogeneous fraction of
degree~$-n-1$. We define~$\rho_0\eqdef[a\omega]_\tGD$ and~$\rho_{k+1} \eqdef [\delta(\rho_k)]_\tGD$.
Since~$\delta$ commutes with~$\Df$, it is clear that~$\rho_k\equiv\delta^k(a\omega)$ modulo~$\Df(T^n)$.
Hence Theorem~\ref{thm:gd-red} implies that~$\rho_k=[\delta^k(a\omega)]_\tGD$.
Thus, by Theorems~\ref{thm:expiso} and~\ref{thm:gd-red} and, for~$u_0,\dotsc,u_m$ in~$\K$,
\[ \sum_{k=0}^m u_k \delta^k( a/f^q ) \in \sum_{k=0}^n \partial_k A_f \text{ if and only if } 
  \sum_{k=0}^m u_k \rho_k = 0. \]
This leads to Algorithm~\ref{algo:ctgd}. 

\begin{proposition}
  Algorithm~\ref{algo:ctgd} applied to a fraction~$R$ satisfying the regularity
  assumption terminates and outputs the Picard-Fuchs equation of~$\oint R$.
\end{proposition}
\begin{proof}
  Correctness has just been proven. Termination follows from
  Theorem~\ref{thm:gd-red}, point~(ii), which implies that the~$\rho_i$ lie in
  a finite-dimensional space, so they are linearly dependent.
\end{proof}

If Conjecture~\ref{conj:dimca} were proven, it would be enough to
replace~$[\,]_\tGD$ by~$[\,]_{n+1}$, or its efficient variant~$[\,]'_{n+1}$, in
Algorithm~\ref{algo:ctgd} to obtain an algorithm which provably outputs the
Picard-Fuchs equation of a rational integral in the singular case.  While
assuming this conjecture gives good results in practice, the absence of a proof
is embarrassing.

It is worth mentioning the treatment of singular cases by a generic
deformation: to compute a differential for~$\oint R$, for some~$R=a/f$, we may
change~$R$ into
\[ R_\lambda = \frac{a}{f+ \lambda \sum_{i=0}^n x_i^{\deg f}}, \]
where~$\lambda$ is a free variable. The denominator of~$R_\lambda$ always
satisfy the smoothness hypothesis, so Algorithm~\ref{algo:ctgd} applies,
over~$\K(\lambda)$, and gives the Picard-Fuchs equation of~$\oint R_\lambda$,
say~$\cL$ in~$\K(\lambda)\langle\delta\rangle$.  Then~$(\lambda^a
\cL)_{|\lambda=0}$, where~$a$ is the unique integer which makes this evaluation
neither zero nor singular, is a differential equation for~$\oint R$.  This
method achieves a good computational complexity, that is polynomial
complexity with respect to the \emph{generic size} of the
output~\cite{BosLaiSal13}, but its practical efficiency is terrible because
most Picard-Fuchs that are interesting to compute are much smaller than the
generic Picard-Fuchs equation.

Another approach, using the reductions~$[\,]_r$, is to loop over~$r$.  We begin
by fixing~$r$ to an initial value, for example~$1$, and we introduce another
variable~$M$, a positive integer.  Then we compute~$\rho_0$, $\rho_1$, etc. as
in Algorithm~\ref{algo:ctgd} but replacing~$[\,]_\tGD$ by~$[\,]_r$, up
to~$\rho_M$.  If there is no linear dependency relation between the~$\rho_k$
then we increase both~$r$ and~$M$ and repeat the procedure.  At some point, the
parameter~$r$ will exceed~$\Cf$ and~$M$ will exceed the order of the
Picard-Fuchs equation of~$\oint R$.  There, a relation will be found between
the~$\rho_k$ and it will give the Picard-Fuchs equation.  It is possible that a
relation is found before the condition~$r \geq \Cf$ is met: it gives of course
a differential equation, but it need not be the minimal one.

Theorem~\ref{thm:einfsupp} and its corollary allow for an interesting variant
of this approach.  As above, we loop over~$r$.  For a given value of~$r$, the
forms~$\rho_0$, $\rho_1$, etc. are computed as in Algorithm~\ref{algo:ctgd} but
using~$[\,]_r$ instead of~$[\,]_\tGD$. Contrary to the previous approach, the
number of~$\rho_i$ we compute before moving to the next value of~$r$ is not
bounded \emph{a priori}. Instead, we compute~$\rho_0$, $\rho_1$, etc. as long
as~$\rho_k$ stays in~$F_n T^{n+1}$.  Since~$F_n T^{n+1}$ is finite dimensional,
we have the following alternative: either there exists a relation between
the~$\rho_k$, or there exists a~$k$ such that~$\rho_k$ is not in~$F_n T^{n+1}$.
In the first case, the relation gives a differential equation for~$\oint R$. In
the second case, we increase~$r$ and start over the computation of
the~$\rho_k$'s.  Corollary~\ref{coro:degen} assures that as soon as~$r \geq
\Cf$, the second condition is never met, so a relation will eventually be
found.  Algorithm~\ref{algo:ct} details the procedure.

\begin{algo}
\centering
\begin{algorithmic}
  \AlgoInput $a/f^q$ a homogeneous rational function in~$\K(\mathbf x)$ of degree~$-n-1$
  \AlgoOutput $\cL \in \dop$ a differential equation for~$\oint R$
  \Procedure{PicardFuchs}{$a/f^q$}
    \For{$r$ from 1 to~$\infty$} 
      \State $\rho_0 \gets [a\:\omega]_r$
      \Comment Compute the subspaces~$X_r^q$ as they are needed.
      \For{$m$ from~$0$ to~$\infty$ while $\deg \rho_m \leq n \deg f$}
        \If{$\mop{rank}_\K(\rho_0,\dotsc,\rho_m) = m+1$}
          \State $\rho_{m+1} \gets [\delta(\rho_{m})]_r$ 
        \Else
          \State compute~$a_0, \dotsc, a_{m-1} \in \K$ such that~$\sum_{k=0}^{m-1} a_k \rho_k = \rho_m$
          \State \Return $\delta^m - \sum_{k=0}^{m-1} a_k \delta$
        \EndIf
      \EndFor
    \EndFor
  \EndProcedure
\end{algorithmic}
\caption{Computation of annihilating operators of the periods of a rational function}
\label{algo:ct}
\end{algo}

\begin{theorem}\label{thm:algo-cor}
  Algorithm~\ref{algo:ct} terminates and outputs an annihilating operator of~$\oint R$.
\end{theorem}

\section{Implementation}\label{sec:impl}

Algorithm~\ref{algo:ct} has been implemented in the computer algebra system
Magma~\cite{BosCanPla97}, with~$\Q(t)$ as base field~$\K$, with the usual
derivation.\footnote{The implementation is available at~\url{http://github.com/lairez/periods}.}
To be able to treat large examples---like the ones in
Section~\ref{sec:cy}---the coefficient swell makes it necessary to
implement a randomized evaluation-interpolation scheme which splits a
computation over~$\Q(t)$ into several analogous computations over different
finite fields. However it comes at a price: since we lack tight \emph{a priori}
bounds on the size of the output---order, degree, size of the
coefficients---the reconstruction step is not certified to be correct,
even though the probability of failure can be made arbitrarily small.  There are
also several ways to cross-check the result independently.  The variant is
described in~\S\ref{sec:eval-interp}.
In the introduction, I mentionned the guessing method which allows, in some
cases, to compute an annihilating operator of a given period but gives no
guarantee about its correctness. The nature of the risk of failure is very
different though. In the evaluation-interpolation method, the algorithm is
randomized and the probability of failure  can be made arbitrarily small.  It
is even less probable that the algorithm returns twice the same wrong result.
It is not possible to fool the algorithm on purpose with a specific input.
In the guessing method, we do not know how to evaluate the risk of failure and
the algorithm is deterministic so an error will be repeated again and again.
It is in principle possible to fool the method with input designed for this purpose.

When a risk of failure is not acceptable, it is possible to compute
certificates which can be used \emph{a posteriori} to prove that what has been
computed is correct, see~\S\ref{sec:comp-cert}.

\subsection{Implementation of~$[\,]_r$ using \gros}
Let~$M$ be the module~$\Omega^{n+1}\oplus\Omega^n$, that is the free module
generated by~$\omega$ and the~$\xi_i$, recall the definitions
in~\S\ref{sec:diffforms}  A convenient way to implement the reduction~$[\,]_r$
is to compute a reduced \gro\footnote{See \cite[chap.~5]{CoxLitOSh98} for
details about \gros\ for modules, the division algorithm, etc.} say~$G$, of the
submodule~$P$ of~$M$ generated by the~$\partial_i f \omega - \xi_i$, that
is~$\udfw\xi_i - \xi_i$.  We choose on~$M$ a monomial ordering,
denoted~$\succ$, such that for all multi-indices~$I$ and~$J$, and all
integer~$j$
\begin{equation}
  |I| +1 \geq |J| + N \Longrightarrow x^I \omega \succ x^J \xi_j.
  \label{eqn:ord-mon}
\end{equation}
For example, any position-over-term (POT) ordering with~$\omega \succ \xi_0
\succ \xi_1 \succ \dotsb$ is fine.  But a term-over-position (TOP),
with~$\omega \succ \xi_0 \succ \xi_1 \succ \dotsb$, extending a graded ordering
on~$A$ works as well. This gives some flexibility in the implementation.
Let~$\mop{rem}_G$ denote the remainder on division by~$G$.  The
condition~\eqref{eqn:ord-mon} on the order is enough to ensure that~$\succ$ behaves
like an order eliminating~$\omega$.

The reason is the following.  If we give to~$\omega$ the degree~$1$ and to
each~$\xi_i$ the degree~$N$, then~$P$ is a homogeneous submodule of~$M$.  Thus
any reduced \gro\ $G$ of~$P$, whatever the monomial order, contains only
homogeneous elements and the remainder on division by~$G$ of a homogeneous
element of degree~$d$ is homogeneous of degree~$d$.  In particular we have
the
\begin{lemma}
  Let~$\alpha$ be an element of~$\Omega^{n+1}$.  Then the coefficient of~$\omega$
  in~$\mop{rem}_G\alpha$ is zero if and only if~$\alpha\in\udfw\Omega^n$.  In
  this case~$\alpha = \udfw\mop{rem}_G\alpha$.
  \label{lem:remG}
\end{lemma}
\begin{proof}
  By definition of~$G$ there exist polynomials~$c_i$ such that
  \[ \alpha = \mop{rem}_G(\alpha) + \sum_{i=0}^n c_i (\udfw\xi_i - \xi_i). \]
  If the coefficient of~$\omega$ in~$\mop{rem}_G(\alpha)$ is zero
  then~$\mop{rem}_G(\alpha)$ is in~$\Omega^n$.
  Identifying the components gives
  \[ \alpha = \udfw \sum_{i=0}^n c_i \xi_i
    = \left(\sum_{i=0}^n c_i\partial_i f\right)\omega
    \quad\text{and}\quad \mop{rem}_G(\alpha) = \sum_i c_i\xi_i. \]
  
  Conversely, assume that~$\alpha=\udfw\beta$, for some~$\beta$
  in~$\Omega^n$.  We may assume that~$\alpha$ is homogeneous of degree~$d$ and
  that~$\beta$ is homogeneous of degree~$d-N$.  In particular~$\alpha - \beta$
  is in~$P$ and~$\mop{rem}_G(\alpha - \beta) = 0$, since~$G$ is a \gro\
  of~$P$.  By linearity~$\mop{rem}_G(\alpha)$ equals~$\mop{rem}_G(\beta)$.
  
  For the grading introduced above, the element~$\beta$ is homogeneous of
  degree~$d-n$, thus so is~$\mop{rem}_G(\beta)$.  Furthermore, the leading
  monomial of~$\mop{rem}_G(\beta)$, with respect to~$\succ$, is at most the leading
  monomial of~$\beta$, which has the form~$x^I \xi_i$ with~$|I| =
  d-N-n$.  The claim follows since no monomial of the form~$x^J \omega$ has
  degree~$d-n$ (with the alternative grading) and is less than~$x^I \xi_i$,
  thanks to hypothesis~\eqref{eqn:ord-mon}.
\end{proof}

In the same way we prove that
\begin{lemma}
  The intersection~$G\cap \Omega^{n}$ is a \gro\ of~${\rm Syz}$.
  \label{lem:gbSyz}
\end{lemma}
Together with a \gro\ of~${\rm Syz}'$, this \gro\ can be used to compute a basis of~$\Syz_q/\tSyz_q$ in the following way.
Using the \gros, we compute the set
\[ S \eqdef \left\{ \mop{lm}(\alpha) \st \alpha\in\Syz_q \right\} \setminus \left\{ \mop{lm}(\alpha) \st \alpha\in\tSyz_q \right\}. \]
Then, for each element~$\alpha$ of~$S$ we pick an element of~$\Syz_q$ whose leading monomial is~$\alpha$.
Those elements form a basis of~$\Syz_q/\tSyz_q$.

\begin{algo}
\centering
\begin{algorithmic}
  \State \emph{Input} --- $\alpha$ an element of~$T^{n+1}$ and~$q$ an integer
  \State \emph{Output} --- $\redgd_q(\alpha)$ as defined in~\S\ref{sec:redgd}  
  \Procedure{RedStep}{$\alpha$, $q$}
    \State $\alpha' \gets \alpha - \alpha_q$
    \State $\rho + \beta \gets \mop{rem}_G(\alpha_q)$, with~$\rho\in\Omega^{n+1}$ and~$\beta\in\Omega^n$
    \State {\bf return} $\alpha' + \rho + \ud\beta$
  \EndProcedure
\end{algorithmic}
\begin{algorithmic}
  \State \emph{Input} --- $r \geq 1$ and~$q \geq 0$ integers
  \State \emph{Output} --- a basis of~$X_q^r$, as defined in~\S\ref{sec:fastcomp}  
  \Procedure{BasisX}{$r$, $q$}
    \If{$r = 1$}
    \State {\bf return} $\left\{ \ud\beta \st \beta \in \text{(a basis of~$\Syz_{q-1}/\tSyz_{q-1}$)} \right\}$
    \Else
    \State $X \gets \textsc{BasisX}(r-1, q+1)$
    \State {\bf return} $\textsc{Echelon}( 
      \textsc{BasisX}(1, q) \cup \left\{ \textsc{RedStep}(\alpha, q) \in X \st \deg \alpha = qN \right\} )$
    \EndIf
  \EndProcedure
\end{algorithmic}
\begin{algorithmic}
  \State \emph{Input} --- $\alpha$ an element of~$T^{n+1}$, $r$ a positive integer
  \State \emph{Output} --- $[\alpha]'_r$ as defined in~\S\ref{sec:fastcomp}  
  \Procedure{Reduction}{$\alpha$, $r$}
    \State $q \gets \deg \alpha / N$ and $\alpha' \gets \alpha - \alpha_q$
    \State $\rho \gets \mop{rem}(\textsc{RedStep}(\alpha_q, q),\textsc{BasisX}(r, q))$
    \State {\bf return} $\rho_q + \textsc{Reduction}(\alpha'+\rho_{q-1}, r)$
  \EndProcedure
\end{algorithmic}
\caption{Computation of~$[\,]_r$}
\label{algo:redr}
\end{algo}

\gros\ in the module~$M$ can be \emph{emulated} by \gros\ in the polynomial ring~$A$
with two extra variables, say~$u$ and~$v$.  Let~$A'$ be~$A[u,v]$, let~$\omega'$
denote~$u^{n+1}$ and~$\xi'_i$ denote~$u^{n-i}v^{i+1}$.  Let~$M'$ be the~$A$-submodule
of~$A'$ generated by~$\omega'$ and~$\xi'_i$.  Let~$P'$ be the ideal of~$A'$
generated by~$\partial_if \omega' - \xi'_i$ and all the monomials~$u^p v^q$,
with~$p+q=n+2$.  Let~$\varphi$ be the~$A$-linear map from~$M'$ to~$M$
sending~$\omega'$ to~$\omega$ and~$\xi'_i$ to~$\xi_i$.  Finally, let~$G'$ be a
\gro\ with respect to any graded monomial ordering~$\succ'$, say the graded
reverse lexicographic ordering, with~$u \succ v \succ x_0 \succ \dotsb \succ
x_n$.

If~$\succ$, the monomial ordering for~$M$, is the TOP ordering proposed above,
then we have~$\varphi(\mop{rem}_{G'} \alpha) = \mop{rem}_G \varphi(\alpha)$,
and the proof is left to the reader.

The computation of~$X_q^r$ and~$[\,]_r$ is detailed in
Algorithm~\ref{algo:redr}.  The function \textsc{Echelon} takes as input a
finite subset~$S$ of~$T^{n+1}$ and outputs a basis in echelon form
of~$\mop{Vect}(S)$, with respect to the monomial order~$\succ$: that is, a basis~$B$ of~$\mop{Vect}(S)$ such that for all element~$b$ of~$B$, the leading monomial of~$b$ does not appear with a non-zero coefficient in the other elements of~$B$.

\subsection{Evaluation and interpolation scheme}\label{sec:eval-interp}

Let~$h(t) = p/q$ be an element of~$\Q(t)$ such that~$q$ is a monic polynomial.
Let~$d$ be the maximum of~$\deg p$ and~$\deg q$, and~$M$ be the maximum of the
absolute values of numerators and denominators of the coefficients of~$p$ and~$q$.
Given distinct primes~$p_1,\dotsc,p_n$, distinct rational
numbers~$u_1,\dotsc,u_m$ and the evaluations~$a_{i,j} \equiv h(u_j) \pmod{p_i}$,
the fraction~$h$ can be reconstructed given that no~$p_i$ divides the
denominator of some coefficient of~$q$, no~$u_j$ annihilates~$q$,
$\prod_{i=1}^m p_i > 2M$ and~$m > 2d$.  To do so, we first compute~$a_i$
in~$\F_{p_i}(t)$ such that~$a_i \equiv h \pmod{p_i}$, using Cauchy
interpolation~\cite[\S5.8]{GatGer99}.  Then, by the Chinese remainder
theorem, we compute~$A$ such that~$A \equiv h \pmod{\prod_i p_i}$.  And then,
using rational reconstruction~\cite[\S5.10]{GatGer99} to each coefficient
of~$A$, we recover~$h$.  Without \emph{a priori} bounds on~$h$, it is still
possible to try to reconstruct it with the method above.  Assume that we obtain
a result~$h'$, and let~$M'$ and~$d'$ be the analogues of~$M$ and~$d$ for~$h'$.
Under randomness assumptions, the bigger~$\prod_{i=1}^m p_i - 2M'$
and~$m-2d'$ are, the higher is the probability that~$h'=h$.

Any algorithm which inputs and outputs elements of~$\Q(t)$ and which performs
only field operations---addition, multiplication, negation, constant one, zero
test, inversion---in~$\Q(t)$ can be turned into a randomized
evaluation-interpolation algorithm, simply by evaluating the input at~$t=u$
and reducing it in~$\F_p$, for several~$p$ and~$u$, and proceeding to the
computation over~$\F_p$.  Indeed, the execution of the algorithm requires a
finite number of operations, either field operations, which commute with~$\nu$,
or zero test. For generic values of~$p$ and~$u$, these tests yield the same
result on evaluated or unevaluated data.  For specific values of~$p$ and~$u$, a
non-zero quantity can be evaluated to zero, so the computation over~$\F_p$ may
fail or return a result which is not the evaluation of the result of the
computation over~$\Q(t)$. It is important to be able to test that in order to
exclude bad evaluations because the reconstruction process does not handle
possibly wrong evaluations.

The number of evaluation points~$(p,u)$ is chosen, \emph{a
priori} or on-the-fly, so that the reconstruction of the outputs is possible
with high probability of success.  If \emph{a priori} bounds on the output are
known it may be possible to certify the result.  If no bounds are known, then
the evaluation-interpolation algorithm may return a false result, but the
probability of this event can be made arbitrarily small.  This
evaluation-interpolation approach is classical in computer algebra for avoiding the
problem of coefficient swell.

Algorithm~\ref{algo:ct} depends on the derivation~$\delta$, which is not a
field operation, so the conversion to an evaluation-interpolation algorithm is
not completely straightforward.

\subsubsection{Principle}
Let~$u$ be in~$\Q$ and~$p$ be a prime number.  Let~$\nu$ be the partial
function~$\Q(t) \to \F_p$, which consists in evaluating~$t$ in~$u$ and reducing
modulo~$p$. The function~$\nu$ is extended coefficient-wise
to~$\Q(t)[\xx]$,~$\Omega$, matrices, etc.

Let~$f$ be a polynomial in~$\Z[t][\mathbf x]$, and~$\nu(f)$ be its evaluation
in~$\F_p[\xx]$.  We can consider the reductions~$[\,]_r$ associated to~$f$, but
also the \emph{evaluated} reduction, denoted~$[\,]^\nu_r$, associated
to~$\nu(f)$, over~$\F_p$. Given~$\alpha \in T^{n+1}$, and for generic
values of~$p$ and~$u$, the evaluations~$\nu(\alpha)$ and $\nu([\alpha]_r)$ are defined and~$\nu([\alpha]_r) = [\nu(\alpha)]^\nu_r$. 
However, the value of~$\nu(\delta a)$ for some form~$a$ cannot be deduced
from~$\nu(a)$, so Algorithm~\ref{algo:ct} requires an adaptation to fit into an
evaluation-interpolation scheme.

As in Section~\ref{sec:ct}, let~$R = a/f^q$ be a rational  function
in~$\Q(t)$, homogeneous of degree~$-n-1$ with respect to the variables~$\xx$.
Let~$\alpha$ be~$a \omega$.  Once  the value of~$r$ is fixed,
Algorithm~\ref{algo:ct} computes the terms of the sequence~$(\rho_i)_{i\in
\N}$, defined by~$\rho_0 = [\alpha]_r$ and~$\rho_{i+1} = [\delta(\rho_i)]_r$,
until it finds a linear dependency relation between the~$\rho_i$.
For a prime~$p$ and an evaluation point~$u$, can we compute~$\nu(\rho_i)$ using
only operations in~$\F_p$?  The answer seems to be negative, but there are two
ways to circumvent this issue.

The first one is to define~$\rho_i$ to be~$[\delta^i(\alpha)]_r$.  With this
definition, the principle and the halting condition~$\deg \rho_i \leq nN$ of
Algorithm~\ref{algo:ct} remain valid. And given~$\nu(\delta^i(\alpha))$, which
is certainly easy to compute, it is possible in this case to
compute~$\nu(\rho_i)$ using only operations in~$\F_p$.  This approach is
feasible but it becomes terrible if~$i$ reaches high values: indeed, the degree
of~$\delta^i(\alpha)$ is~$\deg\alpha + iN$.

Another approach is to compute the matrix of the linear map, say~$m$, such that
\[ \rho_{i+1} = \rho_i^\delta + m(\rho_i), \]
where~$\rho_i^\delta$ denotes the component-wise differentiation of~$\rho_i$, as opposed
to~$\delta(\rho_i)$ which is~$\rho_i^\delta - f^\delta \rho_i$. Such a linear map
exists and its matrix in a certain basis can be computed by
evaluation-interpolation.

\subsubsection{The matrix of~$\delta$}
\label{sec:matdelta}

Let~$J_r$ be the image~$[T^{n+1}]_r$ of the reduction map~$[\,]_r$.
By construction, the reduction~$[\,]_r$ is idempotent, that is~$[\alpha]_r=\alpha$ for all~$\alpha\in J_r$.
The evaluation-interpolation algorithm relies on the following property of the reduction map~$[\,]_r$:
\begin{proposition}
  The space~$J_r$ is stable under component-wise differentiation.
  \label{prop:stabdelta}
\end{proposition}
\begin{proof}[Sketch of the proof]
  This is a consequence of the fact that~$J_r$ is generated by monomials.  More
  precisely, let~$E$ be the, finite of infinite, minimal
  sequence~$(b_0,\dotsc)$ of monomials of~$T^{n+1}$ which
  generates~$T^{n+1}/\ker [\,]_r$; minimal with respect to the lexicographic
  order on sequences of monomials, where the monomials are compared
  with~$\prec$.  Then~$E$ is a basis of~$J_r$ containing only monomials.
\end{proof}

As a consequence $[\delta(\rho)]_r = \rho^\delta - [f^\delta \rho]_r$, for all~$\rho \in J_r$.

Let~$\mathcal M$ be the least set of monomials of~$T^{n+1}$ such
that~$\mop{Vect}\mathcal M$ contains~$\rho_0$ and is stable under the map~$m
:\rho \mapsto [f^\delta\rho]_r$, and let~$B$ be the matrix in~$\Q(t)^{\mathcal
M\times\mathcal M}$ of the map~$m_{|\mop{Vect}\mathcal M}$ in the
basis~$\mathcal M$.  For generic values of~$p$ and~$u$, the basis~$\mathcal M$,
the matrix~$\nu(B)$ and~$\nu(\rho_0)$ are all computable using only operations
in~$\F_p$, once given~$\nu(f)$, $\nu(f^\delta)$ and~$\nu(\alpha)$.
Once~$\mathcal M$, $B$ and~$\rho_0$ are reconstructed over~$\Q(t)$,
the~$\rho_i$ are easily computed with~$\rho_{i+1} = \rho_i^\delta - m(\rho_i)$, and
the minimal operator~$\cL = \sum_i a_i(t)\delta^i$ such that~$\sum_i
a_i(t)\rho_i = 0$ can be deduced.  It seems to be a good idea to reconstruct~$B$
and~$\rho_0$ over~$\F_p(t)$ and compute~$\cL$ modulo~$p$, and only then to use
several moduli to reconstruct~$\cL$ over~$\Q(t)$.  The full procedure is
summarized by Algorithm~\ref{algo:ctei}.

\begin{algo}
\centering
\begin{algorithmic}
  \AlgoInput $R = a/f^q$ a rational function in~$\Q(t)(\mathbf x)$, homogeneous of degree~$-n-1$ w.r.t.~$\xx$
  \AlgoOutput $\cL \in \dop$ an annihilating operator of~$\oint R$, with high probability
  \Procedure{PicardFuchs}{$a/f^q$}
    \Loop
      \State $p\gets$ random prime number
      \State Compute~$\mathcal M$, $\rho_0$ and~$\mop{Mat}_{\mathcal M} m$, as defined in~\S\ref{sec:matdelta}, over~$\F_p(t)$
      by repeated evaluation of~$t$ and rational interpolation.
      \State Compute~$\rho_0, \rho_1, \dotsc$ over~$\F_p(t)$, with~$\rho_{i+1} = \rho_i^\delta - m(\rho_i)$, until finding a
      relation~$\rho_n + \sum_{i=0}^{n-1} a_i\rho_i = 0$ over~$\F_p(t)$.
      \State Using the Chinese remainder theorem and computations modulo previous
      values of~$p$, try to lift the~$a_i$ in~$\Q(t)$.
      \If{possible}
      \State {\bf return} the lifting.
      \EndIf
    \EndLoop
  \EndProcedure
\end{algorithmic}
\caption{Computation of annihilating operators of the periods of a rational
function, randomized evaluation-interpolation method}
\label{algo:ctei}
\end{algo}

\subsubsection{Estimation of the probability of success}

Let~$\mathcal M$, $\rho_0$ and~$A = \mop{Mat}_{\mathcal M}m$ as in~\S\ref{sec:matdelta}, computed over~$\Q(t)$.  For some~$u$ in~$\Q$ and
some prime~$p$, let~$\mathcal M'$, $\rho'_0$ and~$A'$ be the analogues computed
over~$\F_p$.  It is not hard to check that~$\nu(\ker [\,]_r)$ equals~$\ker
[\,]_r^\nu$, where~$\nu(\ker [\,]_r)$ is the set of all~$\alpha$ in~$\ker
[\,]_r$ such that~$\nu(\alpha)$ is defined.
Let~$\alpha$ be an element of~$T^{n+1}$, whose coefficients are polynomials in~$t$ with integer coefficients.
Do we have~$\nu([\alpha]_r) = [\nu(\alpha)]_r^\nu$?
The fact that~$J_r$ is generated by monomials implies that~$[\alpha]_r$ equals~$\mop{rem}(\alpha,\ker [\,]_r)$, and that~$[\nu(\alpha)]^\nu_r$ equals~$\mop{rem}(\alpha,\ker [\,]^\nu_r)$.
The equality is equivalent to~$\nu(\mop{rem}(\alpha, {\ker[\,]_r})) = \mop{rem}(\alpha, \nu(\ker[\,]_r))$.
A sufficient condition is that the set~$L$ of leading monomials of
elements of~$\ker[\,]_r$ equals the set~$L'$ of leading monomials
of~$\nu(\ker[\,]_r)$.  Since~$\mathcal M$ (resp.~$\mathcal M'$) is the
complement of~$L$ (resp.~$L'$) in the set of all monomials
of~$T^{n+1}$, we obtain
\begin{lemma}
  If~$\mathcal M = \mathcal M'$ then~$A' = \nu(A)$ and~$\rho_0' = \nu(\rho_0)$.
  \label{lem:evaltestM}
\end{lemma}

Let~$P$ be  the probability that~$\mathcal M' = \mathcal M$.
Assume for simplicity that~$\deg\alpha \leq nN$ and that~$J_r$ is
included in~$F_nT^{n+1}$.  Let~$V$ be the subspace~$\ker[\,]_r\cap
F_{n+1}T^{n+1}$ and let~$\mathcal B$ be an echelonized basis of~$V$, formed by
elements of~$T^{n+1}$ whose coefficients are in~$\Z[t]$.  For the above
equalities to hold, it is enough that for all~$b$ in~$\mathcal B$, the
evaluation~$\nu(\mop{lc} b)$ of the leading coefficient of~$b$ is not zero.

Under the assumption, somewhat excessive, that for random~$p$ and~$u$
the~$\nu(\mop{lc} b)$, with~$b\in\mathcal B$ are independent and uniformly
distributed in~$\F_p$, the probability~$P$ equals~$(1-\frac{1}{p})^{\#
\mathcal B}$. Of course~$\# \mathcal B \leq \dim F_{n+1}T^{n+1}$ and
\[ \dim F_{n+1}T^{n+1} = \sum_{q=0}^{n+1} \binom{qN-1}{n} \leq \frac{(n+3/2)^{n+1} N^n}{(n+1)!}. \]
So that
\begin{equation}\label{eqn:bound-prob}
  P \geq \left( 1-\frac{1}{p} \right)^{\tfrac 54 e^n N^n} \geq \exp\left( -\frac{5 e^n N^n}{2p} \right).
\end{equation}
So we will choose~$p$ significantly bigger than~$e^n N^n$ to have~$P \ll 1$.
The set~$\mathcal M$ is not computed, so it is not possible to
compare it with~$\mathcal M'$. However, we can compare the different~$\mathcal
M'$ obtained for different values of~$p$ and~$u$. Typically, most of them will
be mutually equal---and hopefully equal to~$\mathcal M$---and a few will
differ. We simply drop the pairs~$(p,u)$ giving degenerated specialisation~$\mathcal M'$.

\subsection{Computing partial certificates}
\label{sec:comp-cert}

Recall that if~$\cL\in \dop$ is an annihilating operator of~$\oint a/f$, a certifate for~$\cL$
is a sequence~$C_0,\dotsc,C_n$ of rational functions in~$\K[\xx,\frac{1}{f}]$ such that
\[ \cL(a/f) = \sum_{i=0}^n \partial_i C_i. \]
As already mentioned, a certificate is desirable because is allows to check
\emph{a posteriori} in a simple way that~$\cL$ annihilates~$\oint a/f$,
idependently of the algorithm used to obtain~$\cL$.
However, a certificate is typically huge~\cite[Rem.~11]{BosLaiSal13} 
and computing a one is necessarily very costly.
A compromise is possible: we may compute a certificate for each reduction~$\rho_k$,
as a~$\beta_k\in T^n$ such that
\begin{equation}
  \rho_k =
  \begin{cases}
    \alpha + \Df\beta_0 & \text{if~$k=0$} \\
    \delta(\rho_{k-1}) + \Df\beta_k & \text{if~$k \geq 1$}.
  \end{cases}
  \label{eqn:certpart}
\end{equation}

Thus, to check that the output~$\cL = \sum_{k=0}^n a_k \delta^k$ of
Algorithm~\ref{algo:ctei} annihilates~$\oint a/f$, it is enough to check
Equation~\ref{eqn:certpart} for~$k\leq r$ and to check that~$\sum_k a_k \rho_k
= 0$.  The first checks imply that~$\rho_k \equiv \delta^k \alpha$ modulo~$\Df(T^n)$,
and the last one implies that~$\cL(\alpha)\in \Df(T^n)$, and thus that~$\cL$
annihilates~$\oint a/f$.  Since the~$\rho_k$'s are in~$F_{n}T^{n+1}$,
the~$\beta_k$ are in~$F_{n+r}T^{n}$ which ensures that their size is kept
reasonnable.

It is possible to modify Algorithm~\ref{algo:ctei} to compute these
certificates~$\beta_k$.  With the notations of~\S\ref{sec:matdelta}, it amounts
to compute~$\beta_0\in T^n$ such that~$\alpha = \rho_0+\Df\beta_0$, and to
compute some~$\gamma_\mu\in T^n$, for~$\mu\in {\mathcal M}$, such
that~$m(\mu) = f^\delta \mu + \Df(\gamma_\mu)$.  Since~$\rho_{k-1} =
\delta(\rho_k) + m(\rho_k)$, it is possible to compute the~$\beta_k$'s as
linear combinations of the~$\gamma_\mu$'s.

In the evaluation-interpolation scheme, it is possible to compute~$\beta_0$ and
the~$\gamma_\mu$'s over~$\F_p$, to reconstruct them over~$\F_p(t)$, then to
compute the~$\beta_k$'s over~$\F_p(t)$ and to reconstruct them over~$\Q(t)$.
Of course, it comes at an additional cost but a preliminary implementation seems
to show that this cost is reasonnable.

\section{Application to periods arising from mirror symmetry}
\label{sec:cy}

Batyrev and Kreuzer~\cite{BatKre10} have recently constructed a family
of~210 smooth Calabi--Yau varieties of dimension three with Hodge
number~$h^{1,1}$ equal to one. Their method is based on toric varieties of
reflexive polytopes.  To each variety is associated a one-parameter mirror family
of varieties and we look for the Picard-Fuchs equation of a distinguished
principal period.  This computation is the first step toward the computation of other important
invariants, like, mirror maps, \emph{instanton} numbers, etc\footnote{For
  an introduction to the topic, see \cite{CoxKat99,BatStr95}.}.  The~210 varieties
gather together into~68 different classes of diffeomorphic
manifolds~\cite[table~3]{BatKre10}.  The principal periods associated to
diffeomorphic varieties need not coincide but they are typically expected to
differ only by a rational change of variable.

In concrete terms, we look for a differential equation satisfied by periods of
rational integrals in the form
\begin{equation}\label{eqn:cyper}
  F(t) \eqdef \oint_\gamma \frac{1}{1-t g(x_1,\dotsc,x_4)} \frac{\ud x_1}{x_1}\frac{\ud x_2}{x_2}\frac{\ud x_3}{x_3}\frac{\ud x_4}{x_4},
\end{equation}
where~$g$ is a Laurent polynomial and the integral is taken over the
cycle~$\gamma$ defined by~$|x_i|=\varepsilon$, with~$\varepsilon$ a small
positive real number.  Here~$g$ is~$\sum_v x^v$, where the sum ranges over the
vertices of a reflexive lattice polytope.  For the~210 polytopes under
consideration, \citeauthor{BatKre10} claim that~$F(t)$ satisfies a linear differential
equation of order~$4$, as a consequence of~$h^{1,1}$ being~$1$.  Moreover,
this differential equation should have maximally unipotent monodromy
at~$t=0$.

A power series expansion of the integrand with respect to~$t$ shows that
\begin{equation}
  F(t) = \sum_n \mop{ct}(g^n) t^n,
  \label{eqn:period-ctseq}
\end{equation}
where~$\mop{ct}(g^n)$ stands for the constant term of~$f^n$.
\Citeauthor{BatKre10} have computed Picard-Fuchs operators for
topologies~\#37,~\#40 and~\#43--68 of their list. They used the \emph{guessing} method
presented in the introduction: they computed the power series expansion
of~$F(t)$, using equation~\eqref{eqn:period-ctseq}, until they reached a
degree~$d$ such that they could find a non-zero solution to the equation
\[ \left( \sum_{i=0}^4\sum_{j=0}^d a_{i,j} t^j \theta^i \right) \cdot F(t) = \cO(t^{5(d+1)+1}). \]
The issue with this technique is not the reconstruction step which can be done
efficiently---with respect to the size of the computed operator---but the
computation of the power series expansion: the number of monomials in~$g^k$
is~$\Theta(k^4)$, so the computation of~$N$ terms of~$F(t)$ with this technique
take~$\Theta(N^5)$ operations in~$\Z$, and we may add an order of magnitude to
reflect the binary complexity.

\Citeauthor{Met12}~\cite{Met12} computed four more equations for
topologies~\#24, \#38, \#39 and~\#41.  His method is also guessing, with
modular evaluation techniques, but he managed to improve the space complexity,
not the time complexity though, in the power expansion step and he provided an
implementation optimized with \textsc{Gpu} programming.  Moreover,
\citeauthor{Alm10}~\cite{Alm10} reports that Straten, Metelitsyn and Sch\"omer have
computed one operator for the topology~\#17.  To the best of my
knowledge, no other computation succeeded in the remaining topologies (\#1--16,
\#18-23 , \#25--36, \#42).

With the implementation described in Section~\ref{sec:impl}, I have been able
to compute a differential equation for the~136 remaining integrals, associated
to~35 different topologies.\footnote{The results are available at~\url{http://pierre.lairez.fr/supp/periods}.}

\subsection{Minimal equation and crosschecking}

The equations obtained from the algorithm are not always minimal, for two
reasons.  Firstly they were obtained with~$r=2$ but a higher value might have
caught a lower order equation.  Secondly, the algorithm computes an annihilating
operator of all the periods of a given rational function; a period associated to a given
cycle may satisfy a lower order equation.

Nevertheless, once any differential equation~$\cL$ for~$F(t)$ is obtained, it
is easy to compute efficiently thousands of terms of its power series
expansion: the relation~$\cL(F)=0$ translates into a linear recurrence relation
on the coefficients of the power series expansion and the initial conditions
are given by Equation~\eqref{eqn:period-ctseq}.  Thus we may try to reconstruct
the minimal equation~$\cL_0$.  By contrast to the guessing method, the
reconstructed equation~$\cL_0$ can be proven correct: it is enough to check
that it is a right divisor of~$\cL$, and that it annihilates the first few
terms\footnote{Up to the maximal integral root of the indicial polynomial at zero
of the right quotient of~$\cL$ by~$\cL_0$.} of~$F(t)$.  If the power series expansion does not
reveal a lower order differential equation, we may conjecture that~$\cL$ is
minimal.  Proving it may be done using methods by van~Hoeij~\cite{Hoe97},
see~\S\ref{sec:topo27} for an example.

Since Algorithm~\ref{algo:ctei} is randomized, it is desirable to have
criteria to crosscheck the result.  The Picard-Fuchs equations of periods of
rational integrals are known to have strong arithmetic properties: regular
singularities with rational exponents and nilpotent $p$-curvature for all
prime~$p$, with a finite number of exceptions~\cite{Kat70}.  Checking these
properties is a good confirmation of the correctness of the output: these
properties are so strong that a bad reconstruction would most probably break
them.  In addition, the computation of many terms of the power series expansion
of~$F(t)$ using a annihilating operator~$\cL$ can also be used as a
crosschecking: if the coefficients computed are all integers, as expected in
view of Equation~\eqref{eqn:period-ctseq}, this is also strong indication that
the operator is indeed correct.

\subsection{Description of the results}

In depth treatment is a work in progress with Jean-Marie~Maillard.  This section
presents two examples.\footnotemark
\footnotetext{There are two numberings.  The first one, used in Table~3 of~\cite{BatKre10}, numbers the~68 different topologies, ordered by
  increasing~$h^{1,2}$ number, covering the~210 smooth Calabi-Yau threefolds
  with Picard number~1.  The second one, used in the database
  \url{http://hep.itp.tuwien.ac.at/~kreuzer/math/0802}, numbers in the
  form~v$x.y$ the~198849 reflexive 4D polytopes satisfying an extra property.
  The letter~$x$ indicates the number of vertices.}

\subsubsection{Topology \#42, polytope v25.59}
For the period~\eqref{eqn:cyper} with
\begin{multline*}
  g = w x y z+w x y+\frac{1}{w x y}+w x z+\frac{1}{w x z}+\frac{w y}{z}
  +\frac{z}{w y}+w y+\frac{1}{w y}+\frac{1}{w z}+w+\frac{1}{w}\\+\frac{x z}{y}
  +\frac{y}{x z}+\frac{1}{x y}+x z+\frac{1}{x z}+x+\frac{1}{x}+\frac{z}{y}+\frac{y}{z}+y+\frac{1}{y}+z+\frac{1}{z},
\end{multline*}
where the first few terms of the power series expansion are
\[ F(t) = 1+22t^2+204t^3+3474t^4+57000t^5+1031080t^6+19368720t^7+\mathcal O(t^8). \]
I have computed the following Picard-Fuchs equation
\begin{multline*}
 t^3(7t+1)^2(25t-1)^2(2t+1)^3(101t+43)^3(3t+1)^3\partial^4
 \\+2t^2(7t+1)(25t-1)(2t+1)^2(101t+43)^2(3t+1)^2(848400t^5\\+1012956t^4+413041t^3+62473t^2+1819t-129)\partial^3
 \\+t(7t+1)(25t-1)(2t+1)(101t+43)(3t+1)(4627173600t^8+10573386192t^7
 \\+10004988192t^6+5027593832t^5+1423146511t^4+219009622t^3\\+15394840t^2+182234t-12943)\partial^2
 \\+(7t+1)(25t-1)(2t+1)(101t+43)(3t+1)(6169564800t^8+13061530080t^7
 \\+11311205016t^6+5112706620t^5+1268815538t^4+164341135t^3\\+9051543t^2+74605t-1849)\partial
 \\+8t(7t+1)(25t-1)(2t+1)(101t+43)(3t+1)(192798900t^6+375787872t^5
 \\+294032949t^4+116697469t^3+24254991t^2+2406495t+81356),
\end{multline*}
or, with~$\theta = t\partial$, in a form which highlights the maximally unipotent monodromy,
\begin{align*}
  1849\theta^4
  &-43t\theta(142\theta^3+890\theta^2+574\theta+129) \\
  &-t^2(647269\theta^4+2441818\theta^3+3538503\theta^2+2423953\theta+650848) \\
  &-t^3(7200000\theta^4+34423908\theta^3+65337898\theta^2+57379329\theta+19251960) \\
  &-t^4(37610765\theta^4+220029964\theta^3+499781264\theta^2+511393545\theta+194039928) \\
  &-2t^5(\theta+1)(54978121\theta^3+324737370\theta^2+665066226\theta+466789876) \\
  &-t^6(\theta+2)(\theta+1)(185181547\theta^2+915931425\theta+1176131796) \\
  &-1212t^7(138979\theta+413408)(\theta+3)(\theta+2)(\theta+1) \\
  &-64266300t^8(\theta+4)(\theta+3)(\theta+2)(\theta+1).
\end{align*}
This equation satisfies the conditions given by \citeauthor{AlmEncStr10}~\cite{AlmEncStr10} and
it is not in their database~\cite{StrDB}.  The computation took 80~seconds
and 30~megabytes of memory on a laptop.

Note that formula~\eqref{eqn:cyper}, and homogeneization, give a rational
function~$a/f$ with~$f$ of degree~$8$ with respect to the integration
variables.  The change of variables which maps~$x$ to~$1/x$ and~$w$ to~$w/y$
lowers this degree down to~$5$. This improves dramatically the computation
time.  This kind of monomial substitution can be found by random trials and
errors.  Among the substitutions that lead to degree~$5$, some are better than
others in terms of computation time; but this seems hard to predict.

\subsubsection{Topology~\#27, polytope v23.289}\label{sec:topo27}
For the period~\eqref{eqn:cyper} with
\begin{multline*}
 f = \frac{1}{w}+w+\frac{1}{x}+\frac{w}{x}+x+\frac{x}{w}+\frac{1}{y}+\frac{w}{y}+\frac{1}{x y}+\frac{w}{x y}+y+\frac{y}{w}+\frac{x y}{w}\\+\frac{1}{z}+\frac{w}{z}+\frac{x}{z}+\frac{1}{y z}+\frac{w}{y z}+\frac{w}{x y z}+z+\frac{z}{w}+\frac{z}{x}+\frac{z}{w x}, 
\end{multline*}
where the first few terms of the power series expansion are
\[ F(t) = 1+18t^2+138t^3+2070t^4+29040t^5+452610t^6+7308000t^7 + \mathcal O(t^8), \]
I have computed an annihilating
operator of order~6 and degree~29, let us denote it~$\cL_6$, which is too large
to be reproduced here.  The operator is not of order~$4$ and has not maximally
unipotent monodromy. Is it the minimal equation of~$F(t)$?  Van~Hoeij has
proved\footnote{Using methods introduced in~\cite{Hoe97}, personnal communication.} that if~$\cL_6$
admits a right factor of order~4 then the degree of the coefficients of this
factor is at most~88.  Thus, admitting that~$\cL_6$ is indeed an annihilating
operator of~$F(t)$, if the minimal annihilating operator of~$F(t)$ has order~4,
it would have degree at most~88.  Zero being the only solution to the system of
linear equations
\[ \sum_{i=0}^4\sum_{j=0}^{88} a_{i,j} t^j f^{(i)}(t) = \mathcal O(t^{405}), \]
where the unknowns are the~$a_{i,j}$, this shows that the minimal annihilating
operator of~$F(t)$ is not of order~4. The argument holds for orders~1, 2, 3
and~5 with respective degree bounds~10, 16, 45 and~125.  This is rather
surprising since it contradicts the claims of Batyrev and Kreuzer.  The
topology~\#17, polytope~v18.16766, shows the same behavior with a minimal
equation of order~6.  This has been first reported by Almkvist~\cite{Alm10},
referring to a computation by Straten, Metelitsyn and Sch\"omer.  As Almkvist
wrote about topology~\#17, ``this example leaves some doubts about the
reflexive polytopes.'' I can only corroborate.  The remaining operators have
not been studied in depth yet, but it seems that only one of the~137 newly
computed periods has a minimal equation of order~$4$.

\printbibliography

\end{document}